%% file: main.tex
\def\notes{1}
\documentclass[11pt]{article}
\include{header}
\ifnum\notes=1
    
    \newcommand{\rfnote}[1]{{\color{ForestGreen}\footnote{{\color{ForestGreen} {\bf RF:} #1}}}}
    
    \newcommand{\srnote}[1]{{\color{blue}\footnote{{\color{blue} {\bf S:} #1}}}}
    
    \newcommand{\dpnote}[1]{{\color{orange}\footnote{{\color{orange} {\bf DP:} #1}}}}
\else

    \newcommand{\srnote}[1]{}
    \newcommand{\dpnote}[1]{}
    \newcommand{\rfnote}[1]{}
\fi

\renewcommand{\epsilon}{\varepsilon}
\newcommand{\eps}{\epsilon}

\newcommand{\inp}[2]{\left\langle #1, #2 \right\rangle}

\DeclareMathOperator{\round}{\mathsf{round}}

\newcommand{\bin}[1]{\left\langle #1 \right\rangle}

\newcommand{\QueryTime}{{\sf QueryTime}}

\DeclareMathOperator{\Lengths}{LenSupp}
\DeclareMathOperator{\image}{Im}

\newcommand{\Stat}{\textsc{Stat}}
\newcommand{\SQDim}{\mathrm{SQdim}}

\newtoggle{anonymous}
\togglefalse{anonymous}

\title{Computational Complexity in
Property Testing\footnote{The conference version of this paper appeared at SODA 2026 \cite{FerreiraPR26}.}
}

\iftoggle{anonymous}{
\author{}
}{
\author{
Renato Ferreira Pinto Jr.\thanks{Part of this work was done while R.F.\ was a student at the University of Waterloo and visiting Boston University. R.F. is supported by an NSERC Postdoctoral Fellowship, and was supported by an NSERC Canada Graduate Scholarship during the development of this article.} \\ Columbia University \and
Diptaksho Palit\thanks{D.P.\ and S.R.\ were supported by the U.S. National Science Foundation under Grant No.\ 2022446.} \\ Boston University \and
Sofya Raskhodnikova$^\ddagger$ \\ Boston University
}
}
\date{}

\begin{document}
\maketitle

\begin{abstract}
We initiate a systematic study of the computational complexity of property testing, focusing on the relationship between query and time complexity. While traditional work in property testing has emphasized query complexity—often via information-theoretic techniques—relatively little is known about the computational hardness of property testers. Our goal is to chart the landscape of {\em time-query} interplay and develop tools for proving time complexity lower bounds.
Our first contribution is a pair of {\em time-query hierarchy} theorems for property testing. For all suitable nondecreasing functions $q(n)$ and $t(n)$ with 
$t(n)\geq q(n)$,
we construct properties with query complexity $\widetilde{\Theta}(q(n))$ and time complexity $\widetilde\Omega(t(n))$. Our weak hierarchy holds unconditionally,
whereas the strong version—assuming the Strong Exponential Time Hypothesis—
provides better control over the time complexity of the constructed properties.

We then turn to halfspaces in $\mathbb{R}^d$, a fundamental class in property testing and learning theory. We study the problem of approximating the distance from the input function to the nearest halfspace within additive error $\eps$. (The distance approximation problem is known to have roughly the same complexity as tolerant property testing for appropriate setting of parameters.) For the distribution-free distance approximation problem, known algorithms achieve query complexity $O(d/\epsilon^2)$, but run in time $\widetilde{\Theta}(1/\epsilon^d)$. We provide a fine-grained justification for this gap: assuming the (integer) $k$-SUM conjecture, any algorithm must have running time $(1/\epsilon)^{\ceil{(d+1)/2}-o(1)}$. This  fine-grained lower bound yields
a provable (under a well-established assumption) separation between query and time complexity for a natural and well-studied (tolerant) testing problem.
We also prove that any randomized Statistical Query (SQ) algorithm under the standard  Gaussian distribution requires  $(1/\eps)^{\Omega(d)}$ queries if the queries are answered with additive error up to $\eps^{\Omega(d)}$, revealing a fundamental barrier even in the distribution-specific setting.


\end{abstract}

\thispagestyle{empty}
\setcounter{page}{0}
\newpage
{
\setcounter{tocdepth}{2}
\tableofcontents
}
\thispagestyle{empty}
\setcounter{page}{0}
\newpage
\setcounter{page}{1}

\section{Introduction}\label{sec:intro}
We initiate a systematic investigation of computational complexity of property testing, focusing in particular on the relationship between query and time complexity. Property testing \cite{RubinfeldS96,GGR98}—along with related models of tolerant testing and distance approximation \cite{PRR06}—was introduced to study extremely efficient algorithms that run in sublinear time and therefore do not even have time to read the entire input. However, most prior work in the area  analyzes only the {\em query complexity} of property testers, not their {\em running time}. Property testing textbooks (see, e.g., \cite{Goldreich17,BhattacharyyaY22}) also focus on query complexity because we understand it better than time complexity and can, in particular, prove lower bounds on query complexity using information-theoretic arguments. In many cases, property testers are computationally very simple and have similar query and time complexity. However, this is not always the case.
The goal of our work is to understand the landscape of {\em query-time interplay} in property testing and to develop tools for proving computational hardness of property testers.

Our first contribution is time-query hierarchy theorems for property testing. For all appropriate\footnote{In our hierarchy theorems, $t(n) = \Omega(\log^2 n)$. We do not optimize polylogarithmic factors, as they are sensitive to model details we consider insignificant; for instance, reading the input length $n$ already takes $\Omega(\log n)$ time.} nondecreasing functions $q(n)$ and $t(n)$, with $t(n)=\widetilde{\Omega}(q(n))$, where $n$ represents the length of the input,
we exhibit properties with query complexity $\widetilde{\Theta}(q(n))$ and time complexity $\tilde \Omega(t(n))$. These results show 
the existence of properties with any desired query complexity and arbitrarily higher time complexity.
We provide two hierarchy theorems: weak and strong. The weak hierarchy theorem holds unconditionally whereas the strong one assumes the Strong Exponential Time Hypothesis (SETH). The latter provides better control over the time complexity of the constructed properties: the weak hierarchy guarantees time complexity at most $2^{\poly(t(n))}$, whereas the strong version guarantees time complexity at most $\widetilde{O}(t(n)^{1+\gamma})$ for arbitrarily small $\gamma > 0$.

Next we focus on 
a specific, fundamental class of properties: halfspaces in
$\bR^d$. This class has been widely studied in 
property testing
\cite{MatulefORS10,MatulefORS09,MatulefORS10b,BalcanBBY12,Harms19,BlaisFH21,ChenP22} and
PAC learning
\cite{DobkinEM96,KalaiKMS08,DKZ20,DiakonikolasKKT21,BermanMR22,RubinfeldV23,GollakotaKSV24}.
It is one of the simplest classes that exhibit a gap between query and time complexity of known sublinear-time algorithms—specifically,
for tolerant testing and the related task 
of distance approximation. 
In the (distribution-free) distance approximation problem for a property $\cP$ (in our case,  $\cP$ is the set of halfspaces in $\bR^d$), the goal is to estimate the distance from a function $f$ to the nearest $g \in \cP$ with additive error $\epsilon\in(0,1)$ and high probability, given sample access to a distribution $\cD$ over the domain and query access to $f$, where the distance between $f$ and $g$ is $\Pru{x \sim \cD}{f(x) \ne g(x)}$.
The query and time complexity of this problem are nearly the same as for tolerant testing of $\cP$ (with appropriate setting of parameters) \cite{PRR06,PallavoorRW22}.

To approximate the distance to $\cH$, the class of halfspaces in $\bR^d$, one can take a sample from distribution $\cD$ and return the empirical distance,  
that is, the fraction of sample points that need to be changed to make the sample consistent with a halfspace. 
Since $\cH$ has VC-dimension $d+1$, a sample of size 
$s:=O(\frac{d}{\eps^2})$ suffices \cite{VapnikC71,ShalevB14}.
However, no computationally efficient algorithms are known for this task.  The  fastest known algorithm is implied by 
the work of Matheny and Phillips \cite{MathenyP21} who showed how to approximate maximum bichromatic discrepancy for halfspaces in $d$-dimensions in time $\tilde{\Theta}(s + 1/\eps^d)$, which 
allows us to estimate
the empirical distance—and thus the distance to $\cH$—
in time $\tilde{O}(1/\eps^d)$.  This leaves a huge gap between the query complexity and the running time of the known distance approximation algorithms for halfspaces:
$O(\frac{d}{\eps^2})$ vs. $\tilde{\Theta}(\frac 1{\eps^d})$.
While this problem is  NP-hard when the dimension $d$ grows with the size of
the input\footnote{The 
hardness results in
\cite{GuruswamiR09,FeldmanGKP09} are stated for learning, but the hard instances
used imply the same lower bounds for distribution-free distance approximation.}
\cite{GuruswamiR09,FeldmanGKP09}, it is not known whether this gap 
is necessary for any constant dimension $d\geq 3$.

We justify 
this gap under the $k$-SUM conjecture, a widely used assumption for proving conditional lower bounds across 
computer science (see the lecture notes \cite{vassilevska2020lecture9}). We demonstrate that, under the (integer)
$k$-SUM conjecture, the exponential dependence on $d$ is necessary for distribution-free distance approximation to $\cH$: specifically, for all
constant
$d \in \bN$ and $\gamma > 0$, there is no distribution-free
distance approximation algorithm for halfspaces over $\bZ^d$
running in time $(1/\epsilon)^{\ceil{(d+1)/2} - \gamma}$. 
E.g., for $d=4$, our result gives a (conditional) separation of $O(\frac 1{\eps^2})$ vs.\ at least $(\frac 1 \eps)^{3-o(1)}$.
Our fine-grained lower bound yields a provable (under a well-established assumption) separation between query and time complexity for a
fundamental
(tolerant) testing problem.

Our hardness result for distribution-free distance approximation for halfspaces raises a natural question: does the hardness persist for standard, well-structured
input distributions, or does it arise only for carefully constructed worst cases?
We provide evidence that the problem remains hard even for well-behaved distributions by proving a lower bound for Statistical Query (SQ) algorithms—a broad class of algorithms that can only access the input distribution via estimates of expectations of bounded query functions, rather than seeing individual samples directly. Specifically, we show that every SQ algorithm for distance approximation (and thus for agnostic learning) of halfspaces under the standard Gaussian distribution over $\bR^d$ must use time at least $(1/\epsilon)^{\Omega(d)}$ when the dimension $d$ is constant. This unconditional lower bound in the SQ model implies that any faster general algorithm must exploit structure beyond what can be learned from simple expectation estimates alone, echoing the fine-grained distribution-free bound and revealing a fundamental computational barrier in the distribution-specific setting.

To summarize, our time-query hierarchies and time/query separations for concrete problems shed light on fundamental differences between information-theoretic and algorithmic barriers.

\subsection{Motivation: known gaps between query and time complexity}\label{sec:known-gaps}
There are many examples of property testing problems for which all known algorithms have much higher time complexity than query complexity. The first examples appear in the seminal work of Goldreich, Goldwasser, and Ron \cite{GGR98}, which studied multiple graph properties, including bipartiteness, $k$-colorability, $\rho$-clique, $\rho$-cut, $\rho$-bisection, and other graph partition problems. In the table of results on p.\ 667, all properties except  bipartiteness exhibit 
significantly different query complexity and running time upper bounds: while the query complexity is polynomial in the relevant parameters, the running time is exponential. (The query complexity bounds for some of these problems have been improved in subsequent work \cite{FiatR21,BlaisS23,ShapiraS24}, but the bounds on running time remain exponential.) Moreover, \cite{GGR98} argued that if the time complexity of testing $k$-colorability (and other NP-complete problems) is polynomial in $1/\eps$ then NP$\subseteq$BPP, because when $\eps$ is sufficiently small, a property tester solves the exact decision problem. Note that our time-query hierarchies hold for constant $\eps$, so the argument that, for small enough $\eps$, the corresponding property testing problems require solving exact decision problems does not apply.

Next, we mention several problems where the gap between the best upper bounds on
query and time complexity scales with the input length $n$. A textbook
\cite[\S10.3.5]{BhattacharyyaY22} highlights the following example: although
\emph{every} graph property is testable with a number of queries independent of
$n$ on minor-free graphs \cite{NewmanS11}, the running time of the best known
algorithms \cite{NewmanS11,onak2012complexity}, which involve a brute-force
search, has a large dependence on $n$. For testing monotonicity of functions
over partially-ordered domains, \cite{FischerLNRRS02,Ras03} gave an algorithm with
query complexity $O(\sqrt n)$, but no sublinear in $n$ bound on the time
complexity is known. For testing whether the input string has a sufficiently
small period (up to $\frac{\log n}6$), \cite{LachishN11} designed an algorithm
that makes $\poly (\log \log n)$ queries, but runs in time $\poly(n)$.

Recent work on testing and learning decision trees has established one setting
where 
a query-time
gap can be formally justified: \cite{KochST23np} showed that
distribution-free testing whether a Boolean function over $n$ variables is a
size-$s$ decision tree (for $s = O(n)$), for constant distance parameter
$\epsilon$,
is NP-hard by reducing from \textsc{VertexCover} instances of size
$n^{\Omega(1)}$. In contrast, the sample complexity of testing this class
is at most its VC-dimension, which is $O(s\log n)$. Thus, sufficiently
strong superpolynomial time lower bounds for \textsc{SAT} would imply a
query-time separation for this problem.

For tolerant testing and distance approximation,
to our knowledge, the only formal justification for a query-time
gap comes from recent work of \cite{BlancKST23} on junta testing,
which showed the NP-hardness of tolerantly testing $k$-juntas
over $n$-variable Boolean functions in the distribution-free setting; as
noted in \cite{BerettaHK26}, implicit in \cite{BlancKST23}
is an $n^{k-o(1)}$ time lower bound for this problem under SETH, which is larger than the
$O(2^k + \log \binom{n}{k})$ sample complexity upper bound from VC dimension theory. However, many gaps between query and time complexity of
known algorithms remain unjustified. One family of examples comes from work on geometric properties \cite{BermanMR22,EdenGR}: halfspaces in $\bR^d$ and triangles, convexity, and connectedness in $\bR^2$. For connectedness, the gap 
is exponential in
relevant parameters, while for halfspaces for constant dimension $d\geq 3$, triangles, and
convexity, it is polynomial. In fact, \cite{BermanMR22} (a journal version of
\cite{BMR16}) had a distance approximation algorithm for halfplanes with query
complexity $\Theta(1/\eps^2)$ and time $\Theta(1/\eps^3)$, but it turns
out the running time can be improved to $\tilde O(1/\eps^2)$ using techniques in
subsequent work of \cite{MathenyP21}.
This raises the fundamental question: When are the gaps in the query and time complexity intrinsic to the problems and cannot be improved?

\subsection{Overview of our results and techniques}\label{sec:results}

\paragraph{Time-query hierarchy.} Our (unconditional) weak and (conditional) strong time-query hierarchies are stated in \Cref{thm:weak,thm:strong}, respectively. They are inspired by the query complexity hierarchies of Goldreich et al.~\cite{GoldreichKNR12}. 
To establish time–query hierarchies, we construct properties by combining two orthogonal sources of hardness: one component enforces the specified query complexity using a 3CNF property (constructed in~\cite{Ben-SassonHR05}) that requires a linear number of queries to test, but is easy to decide given the entire input, while the other enforces the specified time complexity by encoding a language that is either provably hard to decide (for the unconditional weak hierarchy) or assumed hard under SETH (for the strong hierarchy). To link the hardness of the language to property testing, we map the language through an error-correcting code of Spielman~\cite{Spielman96}, which is efficiently constructable, encodable, and decodable.

The final property stitches these parts together through repetition (to adjust input sizes to match the desired query and time bounds) and concatenation (to ensure both forms of hardness coexist). A key concatenation lemma guarantees that the combined property inherits both the query and time lower bounds while remaining testable within the claimed upper bounds. This modular design allows us to compose query and computational hardness in a clean, general way, yielding rich hierarchies. 

Our hierarchy theorems cover bounds ranging from polylogarithmic to nearly
exponential in the input size, and therefore a precondition for formally stating
and proving those results is a suitable computational model for property
testing. In \cref{sec:prelim}, we develop a model of property testing in
random-access machines (RAMs), which are well-suited to algorithm design and
analysis as well as standard in fine-grained complexity. Our model extends the
logarithmic cost RAM model of Cook and Reckhow \cite{CookR73} and captures both
property testing (with query access to the input) and classic computation; the
logarithmic cost model is important to enable principled reductions between
instances of vastly different input sizes, which is a key component of our
proof.

\paragraph*{Distribution-free distance approximation for halfspaces.}

Our fine-grained $(1/\epsilon)^{\ceil{(d+1)/2}-\gamma}$ time lower bound for
distribution-free distance approximation for halfspaces over $\bZ^d$ under the
$k$-SUM conjecture is stated in \cref{thm:distribution-free-testing}.
The proof builds on well-known connections between
$k$-SUM and geometric degeneracy problems: specifically, the hardness of
detecting whether any $d+1$ given points in $\bR^d$ lie on a non-vertical
hyperplane. Our reduction adapts a standard construction from computational
geometry—embedding $k$-SUM instances as configurations of points that are
affinely dependent iff 
the original instance is a YES instance—and
then carefully modifies this point set to produce a labeled sample distribution
that encodes the same combinatorial structure.

To ensure that distinguishing YES from NO instances reduces to estimating the
distance to the class of halfspaces, the construction augments each point with a
small vertical “witness” gadget: each original point is replaced by a pair of
nearby points just above and below its hyperplane, labeled oppositely. Any
halfspace that fits the data must cut correctly between these pairs, so
approximate distance estimation with fine enough additive error effectively
solves the underlying $k$-SUM instance. By working in the integer grid and
bounding coordinates, the reduction avoids trivial hardness due to input
encoding and matches a realistic RAM model. The result is a nearly tight
time lower bound for distribution-free testing of halfspaces in low dimension under a
standard fine-grained complexity assumption. 

\emph{The cost of tolerance.}
One implication of our result is that, for distribution-free  testing of
halfspaces in low-dimensional space, tolerance is more expensive in terms of
time than queries. Specifically, for constant dimension $d$, the query
complexity of standard testing is $\widetilde \Theta(1/\epsilon)$, while the
query complexity of tolerant testing is $\tilde\Theta(1/\epsilon^2)$ by VC theory and
\cite{BermanMR22}; hence the query gap between standard and tolerant testing is
quadratic up to logarithmic factors. In contrast, standard testing can be done
in $\poly(d/\epsilon)$ time via linear programming, from which our
$(1/\epsilon)^{\Omega(d)}$ tolerant lower bound (under the $k$-SUM conjecture)
establishes an arbitrarily large polynomial separation for time complexity.


\paragraph*{SQ lower bound under the Gaussian distribution.}
The SQ model, introduced by \cite{Kearns98} as a natural restriction of the PAC
learning model of \cite{Valiant84}, allows access to the input $f$ only via
approximations to the expectations $\Ex{q(x, f(x))}$ of bounded query functions
$q$ (see \Cref{section:sq-model} for a formal definition). A sample complexity
lower bound against SQ algorithms serves as evidence of a time complexity lower
bound against general algorithms in the sense that any faster algorithm must go
beyond estimating expectations $\Ex{q(x, f(x))}$ from independent samples. SQ is
attractive because unconditional lower bounds can be proved for this model, yet
it still captures a wide variety of learning algorithms (see \cite{Reyzin20} for
a survey).

Our $(1/\epsilon)^{\Omega(d)}$ lower bound for (randomized) SQ distance
approximation algorithms for halfspaces under the Gaussian distribution is
stated in \cref{thm:sq-lb}. The starting point of our proof is an SQ lower bound
of $d^{\poly(1/\epsilon)}$ for agnostically learning halfspaces under the
Gaussian distribution in the high-dimensional setting where $d \gg
\poly(1/\epsilon)$ \cite{DKZ20}. That work constructed, from packing
number lower bounds on the sphere, sets of Boolean functions which are
correlated with halfspaces, yet hard to learn using statistical queries. Our
proof builds upon their work via two key components: 1)~a packing number lower
bound for the low-dimensional sphere, which we obtain by analyzing the limit
distribution of angles in random packings from \cite{CFJ13} in our precise
asymptotic regime; and 2)~a ``pseudorandom'' Boolean function on $\bR^d$, which
is adversarially constructed to be uncorrelated with the queries of a given SQ
algorithm, and serves as the NO instance in our proof.

\subsection{Related work}\label{sec:prior-work}
\paragraph*{Hierarchies and separations in property testing.}
In addition to the query hierarchies in terms of the input size $n$ by \cite{GoldreichKNR12} that inspired our query-time hierarchy theorems, there is a query hierarchy in terms of the distance parameter $\epsilon$ by \cite{Goldreich19}. It states that ``for essentially every function $q : (0,1] \to \bN$, there exists a property of Boolean functions that is testable'' with $O(q(\Omega(\eps)))$ queries but not with $o(q(O(\eps)))$ queries. We believe our techniques can be used to extend this result to query-time hierarchies in terms of $\eps$.
Additionally, \cite{CanonneG18} 
showed
an adaptivity hierarchy 
for property testing by constructing properties that are easy to test
(using $\widetilde O(k)$ queries) for
algorithms with $k$ rounds of adaptivity, but 
require
$\widetilde \Omega(n/k^2)$ queries 
for algorithms with $k-1$ rounds.
Finally, \cite{LachishNS08} proved a separation between space and query
complexity, showing that for all space-constructible functions $s(n)$ satisfying $\log \log n \le s(n) \le \log n$, there exists
a language decidable in space $O(s(n))$ but requiring $2^{\Omega(s(n))}$ queries to test.

\paragraph*{Computational complexity of learning.} In contrast to property
testing, computational learning theory has historically emphasized
computational efficiency. Indeed, the original definition of PAC learning
\cite{Valiant84} requires that the learner run in polynomial time; and
many known barriers in learning theory are computational in nature, \eg for learning
DNFs (see \cite{KearnsV94}, and the recent work \cite{AlmanNPS25} and the
discussion therein) and for parities
\cite{FeldmanGKP09,GrigorescuRV11,Valiant15,KarppaKK18,ChenSZ25}.

We mention two lines of research that are related to our work. First,
\cite{DecaturGR99} introduced the notion of \emph{computational sample
complexity}---the number of samples required to learn a concept class in
polynomial time---and showed that, under standard cryptographic assumptions,
there exist classes whose computational sample complexity is polynomially larger
than the sample complexity of learning in the information-theoretic sense. 
Follow-up work \cite{Servedio00} strengthened these results and extended them to
attribute-efficient learning, and recently
\cite{BlancKST25} showed similar separations under worst-case assumptions (\ie
NP-hardness) rather than cryptographic ones. Second, \cite{Vadhan17,KothariL18}
explored the conditions for efficient learnability by showing that the
computational complexity of learning (in realizable and agnostic settings) is
characterized by certain measures of \emph{refutation complexity}. It would be
interesting to extend their work toward characterizations of efficient
testability.

\paragraph*{Efficient constant-sample testability.} In recent work, which appeared after
the first version of our manuscript, \cite{DworkT25} studied which properties can be
tested with a constant number of samples using a small circuit. Extending the work of
\cite{BlaisY19}, who characterized constant-sample testability in the absence of
computational constraints, \cite{DworkT25} showed that a property of Boolean functions $f$ can be tested
with a constant number of samples using a small circuit if and only if membership in the
property only depends on the average values of $f$ on the input in each part of a small partition,
which is itself computable by a small circuit.

\subsection{Open questions}\label{sec:open}

One open question is
to resolve the exponential gap in the time upper and lower bounds in our unconditional query-time hierarchy theorem.
A hierarchy of the form $\mathsf{BPTIME}[t(n)^{1+\gamma}] \not\subseteq \mathsf{BPTIME}[t(n)]$ in the RAM model would improve this gap to an upper bound $t(n)^{1+\gamma}$ for the lower bound $t(n)$.
Currently, the best known randomized time hierarchy also has an exponential gap between the time upper and lower bounds (through brute forcing all possible random tape initializations).

In \Cref{sec:known-gaps}, we collected many specific testing problems with
unexplained gaps in query and time complexity. Specifically, for halfspace
distance approximation (and tolerant testing), our results partially justify
this gap for dimension $d\geq 4$. For $d=3$, the problem remains open. It is
open as well (for all constant $d\geq 3$) for the uniform distribution (e.g.,
over a unit cube or  sphere).
For the Gaussian distribution, it is unknown
whether $(1/\eps)^{\Omega(d)}$ time is required for this
problem, or whether our SQ lower bound can be circumvented by some
non-SQ algorithm (in the high-dimensional setting, the $d^{\poly(1/\epsilon)}$
SQ lower bound for agnostic learning
shown by \cite{DKZ20} is supported by a cryptographic hardness
result based on the LWE problem \cite{DKR23}).
%
Can these algorithmic gaps be closed by new, more efficient algorithms, or do they reflect inherent complexity that can be formalized through fine-grained or cryptographic lower bounds?

   


 


\paragraph*{Organization.} The rest of this paper is organized as follows. In
\cref{sec:prelim}, we formalize our computational model and lay out technical
preliminaries for the rest of the paper. In \cref{sec:hierarchies}, we prove our
time-query hierarchy theorems. In \cref{sec:distribution-free-lb-halfspaces},
we prove our fine-grained hardness result for distribution-free distance
approximation for halfspaces. Finally, in \cref{sec:sq-lb} we prove our SQ
lower bound for distance approximation for halfspaces under the Gaussian distribution.

\section{Our model and setup}
\label{sec:prelim}

We use $\bin{\cO}$ to denote the binary representation of $\cO$: for integers, this is the standard base-2 encoding; for other objects, the representation will be specified as needed.

\subsection{Computational model}
\label{sec:comp-model}

To precisely characterize time complexity in property testing, we must 
choose
an appropriate computational model. While oracle machines (\eg \cite{Goldreich17}) suffice for analyzing query complexity, RAM models are better suited for fine-grained time analysis
and yield sharper bounds than Turing machines. (The best known simulation of a time-$t$ RAM by a Turing machine takes time
$t^2$). RAM models are standard in 
fine-grained complexity (e.g., \cite{Williams15}). 

We propose a natural augmentation of the classic logarithmic cost RAM model of \cite{CookR73} for property testing. 
In that model, the program has access to an infinite
sequence of registers $X_1, X_2, \dotsc$, each holding an integer; basic operations on register $X_i$ incur cost\footnote{Formally, the cost should be $\log(\max(2, |X_i|))$, but we omit this
technical detail for simplicity.} $\log(|X_i|)$, where $|X_i|$ is the absolute value of the integer stored at $X_i$; and the \emph{indirect} operations of reading and writing from register $X_{X_j}$ incur additional cost $\log(X_j)$. The machine 
can read from an input tape and print to an output tape.

In the original model, the input tape holds a finite binary 
string,
and the operation $\textsc{Read}(X_i)$ reads the next symbol from the input tape into register $X_i$. To model property testing (see \cref{def:tester}), we
replace the single input tape with two: the \emph{parameter tape} holding $p$, readable bit by bit via the operation \textsc{ReadParam}, and the \emph{input tape} holding a binary string $x$, accessible via arbitrary queries using the operation \textsc{QueryInp}$(X_j, X_i)$, which loads the $X_j$-th bit of $x$ into $X_i$.

To handle property testing, we let $n \define |x|$ and place $p = \bin{n}$
on the parameter
tape,\footnote{We fix the alphabet $\{-1,0,1\}$ and stipulate that the input and
parameter tapes consist of the binary strings $x$ and
$p$, respectively,
followed by the symbol $-1$ in all the remaining tape cells.}
so the algorithm can first read the input length and then query bits of $x$. We define the query complexity of a tester as the maximum number of 
\textsc{QueryInp} calls
on inputs of length $n$.
To recover standard computation (without queries), we place the input on the
parameter tape and leave the input tape empty. Algorithm $\cA$ \emph{computes} a
function $f: \zo^* \to \zo^*$ in time $t$ if, given input $p$ on the parameter
tape and emtpy input tape, it prints $f(p)$ and halts after executing
operations with total time cost at most $t$. In both settings, $\cA$ accepts if
it prints 1, and rejects if it prints 0. Randomized algorithms may fail with
with some probability over their internal randomness.

\cref{table:ram} summarizes the list of operations in the model, along with the
time cost of each operation.
While \cite{CookR73} only included addition and subtraction as the built-in
arithmetic operations in their model, 
noting that multiplication could be simulated but was unnecessary for the algorithms
studied in that paper\footnote{The algorithms literature features
a vast diversity of RAM model variations; see \cite{GrandjeanJ22} for a discussion.}\!\!,
we explicitly model multiplication and division between registers $X_i$ and $X_j$ as
taking
$O(\log n \log\log n) = O(\log^{1+o(1)} n)$ time\footnote{This matches
the complexity in multitape Turing machines \cite{HarveyVDH21}; we do not 
optimize polylogarithmic factors.}
when $\abs*{X_i} \le n$ and $\abs*{X_j} \le n$.
We also include the operation \textsc{FlipCoin} to support randomized algorithms.

\begin{table}[t]
    \begin{center}
        \caption{Operations and their costs in the logarithmic cost RAM model for property testing.}
        \label{table:ram}
        \vspace{.5em}
        \bgroup
        \def\arraystretch{1.25} 
        \begin{adjustbox}{center}
        \begin{tabular}{c|l|c|l}
            \textbf{Category}
            & \multicolumn{1}{c|}{\textbf{Operation}}
            & \textbf{Description}
            & \multicolumn{1}{c}{\textbf{Time cost}} \\
            \hline
            \multirow{3}{*}{\begin{tabular}{@{}c@{}}Input \&\\Output\end{tabular}}
            & $\textsc{ReadParam}(X_i)$ & read next bit of $p$ into $X_i$ & $O(1)$ \\
            & $\textsc{QueryInp}(X_j, X_i)$ & read $X_j$-th bit of $x$ into $X_i$ & $O(\log(X_j))$ \\
            & $\textsc{Print}(X_i)$ & output $X_i$ in binary notation & $O(\log(|X_i|)$ \\
            \hline
            \multirow{6}{*}{\begin{tabular}{@{}c@{}}Arithmetic\\\& Control\end{tabular}}
            & $X_i \leftarrow C$, $C$ any integer & constant assignment & $O(1)$\\
            & $X_i \leftarrow X_j + X_k$ & addition & $O(\log(|X_j| \cdot |X_k|))$ \\
            & $X_i \leftarrow X_j - X_k$ & subtraction & $O(\log(|X_j| \cdot |X_k|))$ \\
            & $X_i \leftarrow X_j \cdot X_k$ & multiplication
                & $O(\log^{1+o(1)}(|X_j| \cdot |X_k|))$ \\
            & $X_i \leftarrow \floor{X_j / X_k}$ & integer division
                & $O(\log^{1+o(1)}(|X_j| \cdot |X_k|))$ \\
            & $\textsc{Jump}(m, X_i>0)$ & jump to line $m$ if $X_i > 0$ & $O(\log(|X_i|)$ \\
            \hline
            \multirow{2}{*}{Indirection}
            & $X_i \leftarrow X_{X_j}$ & indirect fetch & $O(\log(X_j) + \log(|X_{X_j}|))$ \\
            & $X_{X_j} \leftarrow X_i$ & indirect store & $O(\log(X_j) + \log(|X_i|))$ \\
            \hline
            Randomness
            & $\textsc{FlipCoin}(X_i)$ & put uniformly random bit in $X_i$ & $O(1)$
        \end{tabular}
        \end{adjustbox}
        \egroup
    \end{center}
\end{table}

\subsection{Property testing}
 Next, we formalize property testing in our model. 
\begin{definition}[Property, distance]
    For $n \in \bN$, a \emph{property} $\cP_n$ (over inputs of length $n$) is a
    set $\cP_n \subseteq \zo^n$. Denote the property $\{\cP_n\}_{n \in \bN}$
    (over inputs of all lengths) by $\cP$.
Given $n$-bit strings $x$ and $y$, let $\dist(x,y)=\Pru{i\in[n]}{x_i\neq y_i}$ denote the relative Hamming distance between them. For
property $\cP$, define
$\displaystyle    \dist(x, \cP) \define \min_{y \in \cP_n} \dist(x, y) \,.
$
String $x$ is  {\em $\eps$-far} from property $\cP$ if $\dist(x, \cP) \geq \eps$.
\end{definition}

\begin{definition}[Tester]\label{def:tester}
    For a property $\cP$ and a parameter $\epsilon \in (0,1)$, an
    $\epsilon$-tester for $\cP$ is a randomized algorithm $\cA$ such that, for each input length
    $n \in \bN$ and $x \in \zo^n$, when $\cA$ is run with
    $\bin{n}$ on the parameter tape and $x$ on the input tape, $\cA$ accepts
    with probability at least $\frac 23$ if $x \in \cP$ and rejects with probability
    at least $\frac 23$ if $x$ is $\epsilon$-far from $\cP$.
    Tester $\cA$ has \emph{time complexity} 
    $t(n)$ if the maximum time it
    takes to terminate on inputs of length $n$ is $t(n)$. Tester $\cA$ has
    \emph{query complexity} 
    $q(n)$ if the maximum number of queries it makes
    on inputs of length $n$ is $q(n)$.
\end{definition}

\begin{remark}
    \Cref{def:tester} places a lower bound of $\Omega(\log n)$ on the time complexity of most
    nontrivial property testing algorithms, as this is the time required to read the input length
    $n$ (and to read a general bit of $x$). We limit our attention to $\Omega(\log
    n)$-time algorithms because below this threshold, the model choice becomes overly influential and the study too brittle to be meaningful.
\end{remark}

\begin{remark}
In \Cref{sec:hierarchies}, we treat $\epsilon > 0$ as a constant and part of the
problem definition, whereas in \Cref{sec:distribution-free-lb-halfspaces}, the
parameter $\epsilon$ is given to the tester (on the parameter tape).
\end{remark}

\subsection{Notation}

\paragraph*{Asymptotic behavior of functions.}
A function $k : \bN \to \bN$ is \emph{eventually surjective} if
its image omits only finitely many natural numbers.
Given nondecreasing $f : \bN \to \bR$ and function $g : \bN
\to \bR$, $f$ has \emph{slope} $O(g(n))$ if there exists a constant $c \in \bN$ such that $f(n) - f(n-1) \le c \cdot g(n)$
for all but finitely many $n \in \bN$. A set $S \subset \zo^*$ is \emph{almost everywhere (a.e.) nonempty}
if $S \cap \zo^n$ is nonempty for all but finitely many $n \in \bN$.

\paragraph*{Computable and invertible functions.} A function $f : \bN \to \bN$ is computable in
time $t : \bN \to \bR$ if there exists an algorithm $\cA$ such that,
on each $n \in \bN$, when $\bin{n}$ is placed on the parameter tape, algorithm $\cA$
computes $f(n)$ in time $t(n)$.

Given a function $f : \bN \to \bN$, the function $f^\dag : \bN \to \bN$ is a
\emph{right inverse} of $f$ if, for all $n$ in the image of $f$, it holds that
$f(f^\dag(n)) = n$. For nondecreasing $f$, we define its \emph{minimum inverse}
$f^{-1} : \bN \to \bN$ as follows: for each $n \in \bN$, let $f^{-1}(n) \define
n'$ where $n' \in \bN$ is the smallest number satisfying $f(n') \ge n$. Then
$f^{-1}$ is a right inverse of $f$, and we have $f(f^{-1}(n)) \ge n$ and
$f^{-1}(f(n)) \le n$.


\paragraph*{String notation.} 
For a string $x$ and natural numbers $s \le t$, we write $x^t$ for the concatenation of $t$
copies of $x$, and $x[s:t]$ for the length-$(t-s+1)$ substring of $x$ between
indices $s$ and $t$, inclusive.

\section{Time-query hierarchies in property testing}
\label{sec:hierarchies}

In this section, we prove our hierarchy theorems (\Cref{thm:weak,thm:strong}). 

\begin{definition}[Admissible query and time functions]\label{def:admissible-functions}
The set of \emph{admissible query and time functions} is the set of
nondecreasing, unbounded functions $q, t : \bN \to \bN$ (query and time
functions, respectively) satisfying the following conditions\footnote{The
constant $2.01$ could be replaced with $2+\gamma$ for any $\gamma > 0$.
This is also true for the constant $2.01$ appearing in \Cref{thm:weak} and \Cref{thm:strong}.}: $q(n)
\le \min\left\{ n, \tfrac{t(n)}{\log^{2.01} n} \right\}$; $q(n)$ is
eventually surjective and
computable
in time $O(q(n) \log q(n))$; $t(n)$ is computable in time $O(t(n))$; and
$t^{-1}(n)$ is computable in time $O(n)$.
\end{definition}

\begin{definition}[Query and time complexity functions and classes]\label{def:query-time-classes}
Given a property $\cP$ and a constant $\eps\in(0,1)$, let $Q_{\eps,n}(\cP)$ and
$T_{\eps,n}(\cP)$ be functions of $n\in\bN$ that denote the query and
time\footnote{
The query complexity $Q_{\epsilon,n}$ is a well-defined number for each $n$. In contrast,
the time complexity $T_{\epsilon,n}(\cP)$ is defined for uniform algorithms and thus only admits asymptotic 
bounds.
}
complexity of $\eps$-testing $\cP$ on inputs
of length $n$, respectively. Given $\eps\in(0,1)$ and nondecreasing functions
$q(n),t(n)$, let $\QueryTime_\eps(q(n),t(n))$ be the class of properties
$\eps$-testable by algorithms that make $O(q(n))$ queries {\em and} run in time
$O(t(n))$.
\end{definition}

\begin{theorem}[Unconditional weak hierarchy]
    \label{thm:weak}
    Let $\epsilon \in (0, 1)$ be a sufficiently small constant.
    Let $q, t : \bN \to \bN$ be admissible query and time functions
    (as in \cref{def:admissible-functions})
    such that
    $q^{-1}(n)$ is computable in time $O(t \circ q^{-1} (n))$.
    Then there exists a property\footnote{We do not try to optimize the constant $2.01$ in the exponent.
    We remark that a direct diagonalization argument in the log-cost RAM model for
    randomized algorithms would suffice to make the constant $1.01$ instead.}
    $\cP\in \QueryTime_\eps\left(q(n), 2^{t(n)^{2.01} \log^{4.05} n}\right)$ such that
    $Q_{\eps,n}(\cP)\geq \Omega\left(\frac{q(n)}{\log q(n)}\right)$ and
     $T_{\eps,n}(\cP)\geq \Omega\left(t(n)\right).
    $
\end{theorem}

Assuming SETH, we obtain a similar hierarchy, where the time complexity of the properties is more tightly characterized. Specifically, for all suitable $q(n)$ and $t(n)$, there exist properties with query complexity $\widetilde\Theta(q(n))$ and time complexity within a small polynomial factor of $t(n)$.

\begin{theorem}[Strong hierarchy from SETH]
    \label{thm:strong}
    Let $\epsilon \in (0, 1)$ be a sufficiently small constant and $\gamma > 0$ be any arbitrary constant.
    Let $q, t : \bN \to \bN$ be admissible query and time functions
    (as in \cref{def:admissible-functions})
    such that
    $t(n) \le 2^{O\left(\frac{q(n)}{\log q(n)}\right)}$ and
    $\log t(n) \log\log t(n)$ has slope $o(1)$.
    Then, under SETH, there exists a property $\cP\in \QueryTime_\eps\left(q(n), t(n)^{1+\gamma}\log n)\right)$ 
    such that
     $Q_{\eps,n}(\cP)\geq \Omega\left(\frac{q(n)}{\log q(n)}\right)$ and 
     $T_{\eps,n}(\cP)\geq \Omega\left(\frac{t(n)}{\log^{2.01} n}\right).$
\end{theorem}

The condition that $\log t(n) \log\log t(n)$ has slope $o(1)$ can be viewed as a pointwise analogue of $t(n) \le 2^{o(n/\log n)}$. We argue this requirement is mild and explain its role in our arguments.
The slope condition is satisfied by natural choices such as $t(n) = \polylog n$, $\poly(n)$, $2^{\polylog n}$, $2^{n^\alpha}$ for $\alpha \in (0,1)$, and $2^{n/\log^{1+\gamma} n}$ for $\gamma > 0$. It suffices for $\log t(n) \log\log t(n)$ to be $o(n)$ and concave.
This condition arises because our lower bounds use a repetition scheme that embeds many small instances (of length $k$) into a larger input (of length $n \gg k$), reducing from the small to the large instance. For this to work, the mapping $n \mapsto k(n)$—essentially $\log t(n) \log\log t(n)$—must be surjective. The $o(1)$ slope ensures surjectivity and allows us to efficiently compute and invert this mapping within our reductions.\footnote{In fact, a sufficiently small constant slope would also suffice.}
A similar surjectivity condition appears in \cite{GoldreichKNR12} for query complexity hierarchies. Although a weaker “dense range” requirement could work, it would complicate our reductions (e.g., requiring careful padding of $k$-SAT encodings). Since the current condition already covers typical $t(n)$, we prioritize simplicity.

{\bf Organization of \Cref{sec:hierarchies}.}
\Cref{sec:hard-lang} constructs languages that are hard for randomized
algorithms in our model. \Cref{section:codes} describes codes used to map these
languages to hard properties. In \cref{sec:3cnf-property}, we design a property
with near-maximal query complexity that is still efficiently testable.
Finally, we combine these elements to obtain the main results: we present a general
construction in \cref{sec:hierarchies-proof-const}, and then apply it to prove
Theorems~\ref{thm:weak} and \ref{thm:strong} in \cref{sec:hierarchies-proof}.

\subsection{Hard languages}
\label{sec:hard-lang}

In this section, we construct 
languages that are hard to decide, \ie 
with high randomized time complexity.
In \cref{sec:hard-lang-ram}, we unconditionally prove the existence of such languages in the RAM model.
The best upper bound known for deciding such languages is exponential in the lower bound.\footnote{A randomized time hierarchy theorem would suffice to overcome this shortcoming.}
In \cref{sec:seth-language}, we prove the existence of a language that is hard to decide and has a nearly matching upper bound, assuming SETH.

\subsubsection{Hard language in the RAM model}
\label{sec:hard-lang-ram}

In this section, we construct a language that is hard to decide in the randomized RAM model.
Our construction proceeds via diagonalization for the classical Turing machine (TM) model.
To translate between computation models, we use 
\cite[Theorem~2]{CookR73}, stated next.

\begin{theorem}[\cite{CookR73}]
    \label{thm:ram-tm-overhead}
        If a language $A$ is decided
        by a RAM within time $T(n) > n$, then $A$ is decided by some multitape TM within time $T(n)^2$.
        Conversely, if a TM decides $A$ within time $T(n) \geq n$, then some RAM decides $A$ within time $T(n) \log T(n)$.
\end{theorem}

\Cref{thm:ram-tm-overhead} is stated only for deterministic machines. However, since coin tosses have $O(1)$ cost in both the TM and RAM model, the theorem extends to randomized computation as well.
Next, we use \cref{thm:ram-tm-overhead} to construct a hard language in the RAM model.

\begin{theorem}
    \label{thm:hard-lang-ram}
    Given a function $t(n) > n$ that is computable in time $O(t(n))$, there exists a language $L$ that is a.e.\ nonempty, which cannot be decided by a randomized RAM within time $O(t(n))$, but can be decided by a deterministic RAM in time $O\Big(2^{t(n)^{2.01}}\Big)$. 
\end{theorem}

\begin{proof}
    
    Let $L$ be the language decided by the following procedure: 
    
    {\em ``On input  $x$ of length $n$, where $x=\langle M\rangle$ for some probabilistic TM $M$,
    \begin{itemize}
        \item Simulate $M(x)$ for $t(n)^{2.005}$ steps 
        on all 
        random 
        branches.
        \item If $M$ halts on all branches and outputs bit $b$ on at least $\frac 23$ of them, output $\overline{b}$; else, output 0''\!\!.
    \end{itemize}}
    The language $L$ cannot be decided by an $O(t(n)^2)$-time probabilistic TM.
    Moreover, $L$ is a.e.\ nonempty, because beyond a certain length, the TM recognizing the empty language appears at every input length, and the negation of its output is a 1.
    By \cref{thm:ram-tm-overhead},  no $O(t(n))$-time randomized RAM decides $L$.
    However, $L$ is decidable in time $2^{t(n)^{2.005}}$ by a deterministic TM, and thus in time $O\left(t(n)^{2.005} \cdot 2^{t(n)^{2.005}}\right)=O\left(2^{t(n)^{2.01}}\right)$ by a deterministic RAM.
\end{proof}

\begin{remark}
    \label{rem:hard-lang-ram-io}
    The language $L$ is hard to decide even infinitely often
    for RAMs running in time $O(t(n))$, because once the machine description $\langle M \rangle$ appears in our sequence of inputs, it appears at every input length hence (due to paddability of machine descriptions), and is thus diagonalized against at every subsequent input length.
\end{remark}

\subsubsection{Hard language from SETH}
\label{sec:seth-language}

\begin{assumption}[SETH]
    \label{assumption:seth}
    For all $\gamma \in (0,1)$, there exists $k \in \bN$ such that no randomized
    algorithm decides \textsc{$k$-SAT} instances on $n$ variables in time
    $2^{(1-\gamma) n}$ with error probability at most $\frac 13$.
\end{assumption}

By the Sparsification Lemma \cite{ImpagliazzoPZ01}, we can assume that the \textsc{$k$-SAT}
instances above have only $m = O(n)$ clauses.

\begin{fact}[Application of the Sparsification Lemma {\cite{ImpagliazzoPZ01}}]
    \label{prop:sparse-seth}
    Under \cref{assumption:seth}, for all $\gamma > 0$,
    there exist $k_\gamma \in \bN$ and $C_\gamma > 0$ such that no randomized
    algorithm decides \textsc{$k_\gamma$-SAT} instances with $n$
    variables and at most
    $\floor{C_\gamma n}$ clauses in time $2^{(1-\gamma) n}$ with error probability at
    most $\frac 13$.
\end{fact}


Next, we define a hard language that serves as a key ingredient in our reductions. 


\begin{definition}
    \label{def:hard-language}
    Fix a binary encoding of CNF formulas and a constant $D > 0$, such that for every $k$-CNF formula $\phi$ with $n$ variables
    and at most $m$ clauses, its encoding $\bin{\phi}$ has $D k m \log n$ bits.
    Fix $\gamma > 0$. Let $k_\gamma \in \bN$, $C_\gamma > 0$ be the
    corresponding constants from \cref{prop:sparse-seth}.
    For each $N \in
    \bN$, let $n \define \floor{\frac{N}{D C_{\gamma/2} k_{\gamma/2} \log N}}$ and
    define the language $L^{(\gamma)} \define
    \bigcup_{N \in \bN} L^{(\gamma)}_N$, where 
    $$L^{(\gamma)}_N \define \left\{\bin{\phi}: \phi \text{ is a satisfiable instance of \textsc{$k_{\gamma/2}$-SAT} on $n$ variables and $\floor{C_{\gamma/2} n}$ clauses}\right\}.$$
    This is well-defined because $N$ bits suffice for the encoding by the choice of $D$.
\end{definition}

Combining SETH with the na\"ive upper bound of $2^n \poly(n)$ for $k$-SAT gives the
following.\footnote{Although faster exponential-time randomized algorithms for $k$-SAT exist (e.g., \cite{PaturiPSZ05}), they do not help our argument, as our bounds are only tight up to constant factors in the exponent.
}.

\begin{lemma}[Bounds for hard language]
    \label{lemma:bounds-hard-language}
    For all constant $\gamma > 0$, there exists a constant $A_\gamma > 0$ such that the
    language $L^{(\gamma)}$ is decidable in $O\Big(2^{\frac{N}{A_\gamma \log N}}\Big)$ time.
    Furthermore, under \cref{assumption:seth}, no randomized algorithm decides $L^{(\gamma)}$ in
    $2^{\frac{(1-\gamma) N}{A_\gamma \log N}}$ 
    time with error probability at most $\frac 13$.
\end{lemma}
\begin{proof}
    Let $A_{\gamma} \define (1 - \gamma/
    4) D C_{\gamma/2} k_{\gamma/2}$ for the constants $D, C_{\gamma/2},
    k_{\gamma/2}$ from \cref{def:hard-language}. The na\"ive brute force algorithm decides
    $L^{(\gamma)}$ in time $2^n \poly(n)$, where $n = \floor{\frac{N}{D C_{\gamma/2} k_{\gamma/2} \log N}}$
    is the number of variables in the \textsc{$k$-SAT} instance. For all sufficiently large $N$,
    this upper bound is
    \[
        2^n \poly(n)
        \le 2^{\frac{N}{D C_{\gamma/2} k_{\gamma/2} \log N}} \poly(N)
        = 2^{\frac{(1 - \gamma/4) N}{A_\gamma \log N} + O(\log N)}
        \le 2^{\frac{N}{A_{\gamma} \log N}} \,.
    \]
    On the other hand, under \cref{assumption:seth}, by \cref{prop:sparse-seth},
    any randomized algorithm  deciding $L^{(\gamma)}$ with error probability at most
    $\frac 13$  runs in time at least
    \[
        2^{(1-\gamma/2) n}
        = 2^{(1-\gamma/2) \floor{\frac{N}{D C_{\gamma/2} k_{\gamma/2} \log N}}}
        \ge 2^{(1-\gamma/2) \left(\frac{(1 - \gamma/
        /4) N}{A_\gamma \log N} - 1\right)}
        \ge 2^{(1-\gamma/2) \left(\frac{(1 - \gamma/2) N}{A_\gamma \log N}\right)}
        \ge 2^{\frac{(1 - \gamma) N}{A_{\gamma} \log N}} \,. \qedhere
    \]
\end{proof}

\subsection{Efficiently constructable, encodable, and decodable codes}
\label{section:codes}

Our reductions use error-correcting codes with constant rate and relative
distance in two ways:
\begin{enumerate}
    \item To transform inputs to a hard decision problem into well-separated
    codewords, which serve as inputs for a property that is
    hard to test.

    \item To transform a family of hard property testing problems with an
    efficient \emph{non-uniform} decision procedure into a problem with
    an efficient \emph{algorithm} by encoding advice into the input.
    (The code ensures that distance is preserved.)
\end{enumerate}

Since we care about both query and time complexity, we need codes with efficient
construction, encoding, and decoding. We get this by combining the linear-time
encodable and decodable codes of \cite{Spielman96} with efficient constructions
of expander graphs. Specifically, we use the Zig-Zag construction of
\cite{ReingoldVW00}, which, as noted in \cite{BrakerskiN13}, is linear-time in
the Word RAM model. In our log-cost RAM model, this yields near-linear time.

\begin{theorem}[Adapted from \cite{Spielman96}]
    \label{thm:code}
    There exist constants $r, \delta \in (0,1)$, where $\frac 1 r\in\bN$, and a family $\cE =
    \left\{\cE_n\right\}_{n \in \bN}$ of error-correcting codes such that 
    each 
    $\cE_n : \zo^n \to
    \zo^{n/r}$
    has relative
    distance at least $\delta$ (and rate $r$). Moreover, 
    $O(n \log n)$ time suffices for constructing $\cE_n$, encoding messages of length
    $n$, and decoding codewords of length
    $\frac nr$.
\end{theorem}

In \Cref{thm:code}, we choose $r$ so that $\frac 1 r$ is an integer, without loss of generality, as the theorem in \cite{Spielman96} holds for
all sufficiently small constant $r$.
\Cref{thm:code} is obtained by replacing the Ramanujan graphs in the construction of \cite[Section~5]{Spielman96} with the 
Zig-Zag graphs from \cite{ReingoldVW00}.
Spielman's proof (in \cite[Theorem~19]{Spielman96}) 
requires 
base graphs  with the second eigenvalue arbitrarily small compared to the degree $d$,
a condition satisfied by the Zig-Zag graps, as stated next.

\begin{theorem}[Implicit in \cite{ReingoldVW00}]
    \label{thm:zig-zag}
    For all constant $\epsilon \in (0,1)$, there exist constant $d \in \bN$ and
    an $O(n \log n)$-time algorithm that, given $n \in \bN$, constructs a
    $d$-regular graph on $\Theta(n)$ vertices with second-largest eigenvalue at
    most $\epsilon d$.
\end{theorem}

\subsection{A property requiring $\Omega(n / \log n)$ queries and $O(n \log n)$ time to test}
\label{sec:3cnf-property}

In this subsection, we construct a property $\cQ$ such that testing $\cQ_n$ has nearly maximal query
complexity, but can be done in near-linear time. 
Our construction is based on the result of \cite{Ben-SassonHR05} establishing the existence of 3CNF
properties that are hard to test. They 
exhibit a family $\{\varphi_m\}_m$ of 3CNF
formulas with $O(m)$ clauses on $m$ variables such that testing whether an assignment $x \in \zo^m$
satisfies $\varphi_m$ requires $\Omega(m)$ queries. It is easy to check whether an
assignment satisfies a \emph{fixed} $\varphi_m$ in $O(m \log m)$ time.




\begin{fact}[\cite{Ben-SassonHR05}]
    \label{fact:3cnf-hard}
    Let $S(\varphi_m) \subset \zo^m$ be the set of satisfying assignments for a 3CNF formula $\varphi_m$ on $m$ variables and
$\cS_\varphi \define \{ S(\varphi_m) \}_m$ be the property obtained from a given family
$\{\varphi_m\}_m$ of 3CNF formulas.
    There exist constants $\epsilon_0 \in (0,1)$ and $A > 0$,
    and a family $\{\varphi_m\}_{m \in \bN}$ of 3CNF formulas with at most $A m$ clauses on
    $m$ variables such that $\epsilon_0$-testing $\cS_\varphi$ requires $\Omega(m)$ queries.
\end{fact}

At first glance, the property $\cS_\varphi$ appears to meet our needs.
However, the issue is that the procedure for checking satisfying
assignments is non-uniform: for each $m$, there exists a hard-to-test formula
$\varphi_m$ with a corresponding efficient verifier, yet no single algorithm
works for all~$m$.
The formulas in
\cite{Ben-SassonHR05} are derived from random codes using random expander graphs
with strong unique-neighbor expansion, 
for which no explicit
(let alone 
near-linear time) constructions are known.\footnote{Explicit
constructions of related structures, like lossless expanders (see \eg
\cite{CapalboRVW02}), have been extensively studied. However,
\cite{Ben-SassonHR05} relies on two-sided expansion, which is generally harder
to achieve and remains an active area of research
\cite{ChattopadhyayGRZ24,HsiehMM024,HsiehLMOZ25,Chen25}.
}

To address this, we construct a new property $\cQ$ as follows. For input length
$n$, we define $\cQ_n$ to contain strings formed by concatenating: (1) an
encoding of a 3CNF formula $\psi_m$ with at most $A m$ clauses on $m$
variables, and (2) a satisfying assignment for $\psi_m$. Since encoding $\psi_m$
takes $\Theta(m \log m)$ bits, we choose $m = \Theta(n / \log n)$.
To ensure that the strings obtained via this construction are far from
$\cQ_n$ when the assignment is far from $S(\psi_m)$,
we encode $\psi_m$ as $\cE(\langle \psi_m \rangle)$, where $\cE$ is the code
from \cref{thm:code} and $\langle \psi_m \rangle$ is the binary representation
of $\psi_m$. Because the assignment is only $\Theta(n / \log n)$ bits, we apply
a repetition code to expand it to $\Theta(n)$ bits.
The resulting property $\cQ$ is hard to test (via
\cref{fact:3cnf-hard}) but admits a linear-time tester that simply reads and
decodes $\psi_m$, extracts the assignment, and checks whether it satisfies the
formula.

\begin{definition}\label{def:3cnf-aux}
    For all $m \in \bN$,
    let $\Psi_m$
    be the set of 3CNF formulas on $m$ variables with at most $A m$
    clauses, where $A$ is as defined in \cref{fact:3cnf-hard}.
    Let $D > 0$ be a constant
    such that formulas in $\Psi_m$ can be encoded in $\ell_m \define
    \floor{D m \log m}$ bits.
\end{definition}

\begin{definition}[Property with near-linear query and time complexity]
    \label{def:3cnf-construction}
    Let $r \in (0,1)$ be the rate parameter from \cref{thm:code}.
    Define the property $\cQ = \{\cQ_n\}_{n \in \bN}$ as follows.
    Let $n \in \bN$ be sufficiently large so that $m \define \floor{\frac{nr}{2 D \log n}} \ge 1$. For each
    3CNF formula $\psi_m \in \Psi_m$ with encoding $\langle \psi_m
    \rangle$ of length $\ell_m$, and  each satisfying assignment
    $x \in S(\psi_m)$,
    let $n' \define n - \floor{\ell_m/r}$,
    $t \define \ceil{n' / m}$, and  $y(x) \define x^t[1 : n']$.
    Finally, let
    \[
        \cQ_n \define \left\{ \cE\left(\langle \psi_m \rangle\right) \circ y(x)
            :\; \psi_m \in \Psi_m \text{ and } x \in S(\psi_m)
        \right\}\,,
    \]
    where $\cE$ is the code from \cref{thm:code}. For small $n \in \bN$ yielding $m = 0$, let
    $\cQ_n \define \emptyset$.
\end{definition}

We now show that $\cQ_n$ satisfies our requirements on query and time complexity.

\begin{lemma}[Query lower bound]
    \label{lemma:3cnf-query-bound}
    Let $\epsilon > 0$ be a sufficiently small constant. Then $\epsilon$-testing 
    property $\cQ$ from \cref{def:3cnf-construction} requires 
    $\Omega(n / \log n)$ queries.
\end{lemma}
\begin{proof}
    We reduce 
    testing $\cS_\varphi$ (on inputs of length $m$ obtained from $n$ as in
    \cref{def:3cnf-construction}) to testing $\cQ_n$. Suppose $\cT$ is an
    $\eps$-tester for $\cQ_n$. We give
    a $4\eps$-tester $\cA$ for $\cS_\varphi$ that works as follows. On input $x \in \zo^m$, tester $\cA$
    computes $\cE(\langle \varphi_m \rangle)$ for formula $\varphi_m$ from \cref{fact:3cnf-hard}
    using no queries, and then simulates $\cT$ on $z \define
    \cE(\langle \varphi_m \rangle) \circ y(x)$, using at most one query to $x$ per query of $\cT$ to
    $z$. The query complexity of $\cA$ is at most that
    of $\cT$. So, if $\cA$ is a valid tester for $\cS_\varphi$, then by
    \cref{fact:3cnf-hard}, $\cT$ must make $\Omega(m) = \Omega(n / \log n)$
    queries.

    First, suppose $x \in \cS_\varphi$. Then, by construction, $z \in \cQ_n$, so
    $\cA$ accepts with probability at least $\frac 23$.

    Now suppose $x$ is $4\epsilon$-far from $\cS_\varphi$. We claim that $z$ is
    $\epsilon$-far from $\cQ_n$. Indeed, let $z^*$ be a closest string to $z$ in $\cQ_n$. Then $z^*
    = \cE(\langle \psi^*_m \rangle) \circ y(x^*)$ for some
    $\psi^*_m \in \Psi_m$
    and some satisfying assignment $x^* \in S(\psi^*_m)$. We claim that
    $\dist(z, z^*) \ge \epsilon$,
    and consider two cases to prove the claim. First, suppose $\varphi_m \ne
    \psi^*_m$. Then $\dist(\cE(\langle \varphi_m \rangle), \cE(\langle \psi^*_m \rangle)) \ge \delta$,
    where $\delta$ is the distance of the code from \cref{thm:code}.
    Since $\floor{\ell_m/r} = \Omega(n)$ by \cref{def:3cnf-aux},
    it follows that $\dist(z, z^*) \ge \epsilon$ as long
    as $\epsilon$ is smaller than $\delta$ by a sufficiently small constant factor. Second, suppose
    $\varphi_m = \psi^*_m$. Then $\dist(x, x^*) \ge 4\epsilon$ because $x$ is $4\epsilon$-far from
    $\cS_\varphi$ while $x^* \in \cS_\varphi$, and hence $\dist(y(x), y(x^*)) \ge 2\epsilon$ (the
    factor of $2$ accounts for a possibly partial last repetition of $x$ inside $y(x)$, and
    similarly for $x^*$). But $|y(x)| = |y(x^*)| = n - \floor{\ell_m/r} \ge n/2$, so $\dist(z, z^*) \ge \epsilon$. Hence $z$ is $\epsilon$-far from $\cQ_n$,
    so $\cA$ rejects with probability at least $\frac{2}{3}$.
\end{proof}

\begin{lemma}[Time upper bound]
    \label{lemma:3cnf-time-bound}
     There exists an $O(n \log n)$-time decider (\ie a 0-tester) for~$\cQ$.
\end{lemma}
\begin{proof}
    Let $\cT$ be a tester that, on input $z$,  queries all of $z$, computes $m$ and
    $\ell_m$ from \cref{def:3cnf-aux,def:3cnf-construction}, and
    tries to use the decoder from \cref{thm:code} to compute a decomposition $z
    = \cE(\langle \psi_m \rangle) \circ y(x)$. Tester $\cT$ rejects if decoding
    the first $\floor{\ell_m/r}$ bits fails or does not yield a valid
    formula in $\Psi_m$, or if the remaining bits are
    not of the form $y(x)$ for some $x \in \zo^m$. Otherwise, $\cT$ accepts iff the
    assignment $x$ satisfies the formula $\psi_m$.
    
    Since querying the entire input, decoding the code from \cref{thm:code}, and
    checking whether the assignment $x$ satisfies $\psi_m$ all take $O(n \log
    n)$ time, tester $\cT$ runs in time $O(n \log n)$. 

    We claim that $\cT$ answers correctly with probability $1$.
    Clearly, $\cT$ accepts if $z \in \cQ_n$. Now, suppose $z \notin \cQ_n$.
    Suppose $z$ is of the form $\cE(\langle \psi_m
    \rangle) \circ y(x)$, since otherwise $\cT$ rejects. Then $x
    \not\in S(\psi_m)$ (since otherwise $z$ would be in $\cQ_n$), and hence
    $\cT$ rejects, as desired.
\end{proof}

\subsection{Construction}
\label{sec:hierarchies-proof-const}

\newcommand{\concat}{\mathsf{Concat}}

We start by defining operations on properties and summarizing their features.

\subsubsection{Concatenated properties}\label{sec:concatenated-properties}
To prove the time-query hierarchy theorems, we first construct properties achieving the desired query and time bounds separately. The following definitions and lemmas show how to combine them to obtain both guarantees simultaneously.

\begin{definition}[Concatenation of properties]
    Let $n_1, n_2 \in \bN$. Let $\cP^{(1)}_{n_1}$ and $\cP^{(2)}_{n_2}$ be properties over inputs of length $n_1$ and $n_2$, respectively.
    Then their \emph{concatenation} $\concat(\cP^{(1)}_{n_1}, \cP^{(2)}_{n_2})$ is the property over inputs of length $n_1 + n_2$ given by
    $
        \concat(\cP^{(1)}_{n_1}, \cP^{(2)}_{n_2}) \define \left\{ xy : x \in \cP^{(1)}_{n_1}; y \in \cP^{(2)}_{n_2} \right\}.
    $
\end{definition}

\begin{definition}[YES query provider]
    \label{def:query-providers}
    Let $t : \bN \to \bN$ be a function, possibly sublinear in $n$. A
    $t(n)$-time {\em YES query provider} for a property $\cP$ is an algorithm
    $\cA$ that, given $\bin{n}$ for which $\cP_n \ne \emptyset$, runs in $t(n)$
    time to produce an implicit representation of $x^{(n)} \in \cP_n$, and then answers each query $i \in [n]$
    with $x^{(n)}_i$ in $O(\log^{1+o(1)} n)$ time.\footnote{Formally, we can
    define this via two algorithms, \textsc{Preprocess} and \textsc{Query}:
    \textsc{Preprocess}, given $\bin{n}$ on the parameter tape, runs in $t(n)$
    time and puts $\cO \in \zo^*$ to registers; \textsc{Query}, given $\langle i, \cO
    \rangle$, computes $x^{(n)}_i$ in $O(\log^{1+o(1)} n)$ time, where $x^{(n)} \in
    \cP_n$ is uniquely determined by $\cO$. For clarity, we use the
    ``interactive'' version instead.}
\end{definition}

We summarize features of properties obtained via concatenation using notation
from \Cref{def:query-time-classes}.

\begin{lemma}[Concatenation Lemma]\label{lem:concat}
    Given a.e.\ nonempty
    properties $\cP^{(1)}$ and $\cP^{(2)}$, define property $\cP$  by
    setting
    $\cP_{2n} \define \concat(\cP^{(1)}_{n}, \cP^{(2)}_{n})$ and
    $\cP_{2n-1}\define\emptyset$ for each $n\in\bN$.
    Let $\epsilon \in (0, 1)$. Then the following hold.
    \begin{enumerate}
        \item \label{lem:concat-1}
        $Q_{\eps,2n}(\cP)=O\left(Q_{\eps,n}(\cP^{(1)})+Q_{\eps,n}(\cP^{(2)})\right)$ and $Q_{\eps/2,2n}(\cP)=\Omega\left(Q_{\eps,n}(\cP^{(1)})+Q_{\eps,n}(\cP^{(2)})\right)$.

        \item \label{lem:concat-2}
        $T_{\eps,2n}(\cP)=O\left(\left(T_{\eps,n}(\cP^{(1)})+T_{\eps,n}(\cP^{(2)})\right) \log n\right).$

        \item \label{lem:concat-3}
            Suppose $T_{\eps/2,2n}(\cP) = O(t(n))$. For all $i\in[2]$,
            if $\cP^{(i)}_n$ has an
            $O(t(n) \log^{1+o(1)} n)$-time YES query provider, then
            $T_{\eps,n}(\cP^{(3-i)})=O(t(n) \log^{1+o(1)} n)$.
    \end{enumerate}
\end{lemma}
\begin{proof}
    {\bf \cref{lem:concat-1}.} We prove the upper and lower bounds separately.
    
    Let $\cA_1$ and $\cA_2$ be
    $\epsilon$-testers for $\cP^{(1)}_{n}$ and
    $\cP^{(2)}_{n}$ that have error probability at most
    $\frac 16$.
    Construct the following $\epsilon$-tester $\cA$ for $\cP_{2n}$. On
    input $xy$, where $|x|=|y|=n$, algorithm $\cA$ runs $\cA_1$ on
    $x$ and $\cA_2$ on
    $y$, accepting iff both $\cA_1$ and $\cA_2$ accept.
    We show that this is a valid $\epsilon$-tester for $\cP_{2n}$.
    If input $xy \in \cP_{2n}$,
    then $\cA_1$ and $\cA_2$ accept $x$ and $y$, respectively,
    with probability at least $\frac{5}{6}$ each, so $\cA$ accepts with
    probability at least $\frac{2}{3}$ (by a union bound).
    If input $xy$ is $\epsilon$-far from $\cP_{2n}$, then, by an averaging argument,
    $x$ or $y$ must be $\epsilon$-far from $\cP^{(1)}_{n}$ or
    $\cP^{(2)}_{n}$,
    respectively. Thus, $\cA_1$ or $\cA_2$ rejects (and then so does $\cA$) with probability at
    least $\frac{5}{6}$. 
    By standard probability amplification, there exist algorithms $\cA_1$ and $\cA_2$, as specified above, that make $O(Q_{\eps,n}(\cP^{(1)}))$ and $O(Q_{\eps,n}(\cP^{(2)}))$ queries, respectively. The query complexity of $\cA$ constructed from such $\cA_1$ and $\cA_2$ is $O\left(Q_{\eps,n}(\cP^{(1)})+Q_{\eps,n}(\cP^{(2)})\right)$.

    Given an $\frac{\epsilon}{2}$-tester $\cT$ for $\cP_{2n}$, we construct an
    $\epsilon$-tester $\cT_1$ for $\cP^{(1)}_{n}$ as follows: on input $x \in
    \zo^{n}$, algorithm $\cT_1$ fixes any $y \in \cP^{(2)}_{n}$
    (which exists since $\cP^{(2)}$ is a.e.\ nonempty)
    and simulates $\cT$ on $xy$. If $x \in
    \cP^{(1)}_{n}$, then $xy \in \cP_{2n}$,
    so $\cT_1$ accepts with probability at
    least $\frac 23$. If $x$ is $\epsilon$-far from $\cP^{(1)}_{n}$,
    then $xy$ is
    $\frac{\epsilon}{2}$-far from $\cP_{2n}$, so $\cT_1$ rejects with
    probability at least $\frac 23$. Since the query complexity of $\cT_1$ is at most that of $\cT$,
   we conclude that $\cT$ makes at least
    $Q_{\epsilon,n}(\cP^{(1)})$
    queries. Analogously, it also makes at least
    $Q_{\epsilon,n}(\cP^{(2)})$ queries.
    Thus, $\frac \eps 2$-testing $\cP_{2n}$
    has query complexity
   $\Omega\left(Q_{\epsilon,n}(\cP^{(1)}) + Q_{\epsilon,n}(\cP^{(2)})\right)$.

    \textbf{\cref{lem:concat-2}.} 
    By standard probability amplification, there exist algorithms $\cA_1$ and $\cA_2$, as specified in the previous part, with running time $O(T_{\eps,n}(\cP^{(1)}))$ and $O(T_{\eps,n}(\cP^{(2)}))$, respectively.
    The running time of $\cA$ constructed from such $\cA_1$ and $\cA_2$
    (allowing for simulation overhead, including a potential $\log n$
    blow-up in the cost of queries
    by the combined tester\footnote{Specifically, tester $\cA_2$ may have low-indexed queries to $y$ which are inexpensive, but those queries still cost
    $O(\log n)$ for the tester simulating it.
    While this blow-up does not occur in queries to $x$ by $\cA_1$, we allow a
    global $\log n$ overhead in our upper bound for simplicity. We do not incur
    a $\log t_1$ or $\log t_2$
    overhead for simulating \emph{register} access by either algorithm, since we can
    interleave the registers used by $\cA_1$ and $\cA_2$ in our construction of $\cA$.})
    is $O\left(\left(T_{\eps,n}(\cP^{(1)})+T_{\eps,n}(\cP^{(2)})\right) \log n\right).$

    \textbf{\cref{lem:concat-3}.} Let $\cT$ be an $\frac{\epsilon}{2}$-tester for
    $\cP_{2n}$ running in time $O(t(n))$
    and $\cB$ be an $O(t(n) \log^{1+o(1)} n)$-time YES query provider
    for $\cP^{(2)}_{n}$.
    We obtain an $\epsilon$-tester $\cT_1$ for $\cP^{(1)}_{n}$
    as follows. On input 
    $x \in \zo^{n}$, tester $\cT_1$ uses $\cB$ to fix some $y
    \in \cP^{(2)}_{n}$ (which exists since $\cP^{(2)}$
    is a.e.\ nonempty)
    and simulates $\cT$ on $xy$: when $\cT$ queries a bit
    from $xy$, tester $\cT_1$ either queries a bit from $x$ or uses $\cB$ to query
    a bit from $y$ in $O(\log^{1+o(1)} n)$ time.
    Then, as in Item~1, algorithm $\cT_1$ is
    an $\epsilon$-tester for $\cP^{(1)}$. By construction, the running time of
    $\cT_1$ is $O(t(n) \log^{1+o(1)} n)$. An identical argument works for $\cT_2$.
\end{proof}

\begin{remark}
    \label{rem:lb-type-preserved-concat}
    For \textbf{Item 3} of \cref{lem:concat}, if the tester for $\cP$ succeeds only on infinitely many input lengths (as opposed to the usual notion of almost every input length), then we get testers for $\cP^{(1)}$ and $\cP^{(2)}$ that also succeed on infinitely many input lengths.
\end{remark}

\subsubsection{Intermediate complexity regimes by repetition}\label{sec:repeated-insteances}
We first construct properties with maximal query or time complexity, and then obtain intermediate complexities using  instance repetition, as in \cite{GoldreichKNR12,KelmanLR25}.

\begin{definition}[Repeated instances \cite{GoldreichKNR12}]\label{def:repeated-instances}
    Given a property $\cP$ and a function $k : \bN \to \bN$
    satisfying $k(n) \le n$, define the \emph{repeated instances} property $\cP^{(k)}$ by
        $$\textstyle
        \cP^{(k)}_n \define \left\{ x^r[1:n] : \  x \in \cP_{k(n)} \text{ and } r = \left\lceil \frac{n}{k(n)} \right\rceil \right\}. 
        $$
\end{definition}
We summarize features of properties obtained via repetition using notation from
\cref{def:query-time-classes}.

\begin{definition}
    Given a property $\cP$, we write $\Lengths(\cP) \define \{ n \in \bN : \cP_n \ne \emptyset \}$.
    A set $S \subseteq \bN$ is \emph{decidable} in time $t : \bN \to \bR$ if there is an algorithm
    which, on input $n \in \bN$ given as $\bin{n}$, halts in time $O(t(n))$ and accepts iff $n \in
    S$. Given a function $f$, we write $\image(f)$ for the image 
    of $f$.
\end{definition}

\begin{observation}
    If $f : \bN \to \bN$ is eventually surjective, then $\image(f)$ is decidable in time $O(1)$.
\end{observation}

\begin{lemma}[Repeated instances lemma]\label{lem:rep-inst}
Given a property $\cP$ and
    a function $k : \bN \to \bN$ satisfying $k(n) \le n$
    with a right inverse $k^\dagger$ and
    such that $\Lengths(\cP) \setminus \image(k)$ is finite,
    consider the repeated instances property $\cP^{(k)}$.
    Let $\eps \in (0, 1)$ be a constant.
    The following are true.
    \begin{enumerate}
        \item \label{lem:rep-inst-1} \cite{GoldreichKNR12} If $Q_{\eps/2,n}\left(\cP\right)=q(n)$ then
        $Q_{\eps,n}\left(\cP^{(k)}\right) =O\left(q(k(n)\right)$ and
        $Q_{\eps/4,n}\left(\cP^{(k)}\right) =\Omega\left(q(k(n)\right)$.
        
        \item \label{lem:rep-inst-2}
        If $T_{\eps/2,n}\left(\cP\right)=O(t(n))$
        and $k(n)$ is computable in $O(t(k(n)))$ time, then\\
        $T_{\eps,n}\left(\cP^{(k)}\right)=O(t(k(n)) + \log^{1+o(1)} n)$.

        \item \label{lem:rep-inst-3} Suppose $T_{\eps/2,n}\left(\cP^{(k)}\right)=O(t(k(n)))$.
        For all $\gamma > 0$, if $\image(k)$ is decidable and $k^\dagger(n)$ computable
        both in $O(t(n) \log^{1+\gamma} k^\dagger(n))$ time, then
        $T_{\eps,n}\left(\cP\right)=O(t(n) \log^{1+\gamma} k^\dagger (n))$.
    \end{enumerate}
\end{lemma}
\begin{proof}
    \textbf{\cref{lem:rep-inst-1}.} We prove the upper and lower bounds separately.

    Let $\cA$ be an $\frac{\epsilon}{2}$-tester for $\cP$.
    Construct the following $\epsilon$-tester $\cA'$ for $\cP^{(k)}$.
    First, $\cA'$ repeats the following $O(1/\epsilon)$ times: uniformly select a $j \in [k(n)]$
    and $r \in [(n/k(n)) - 1]$ and check whether $x[r \cdot k(n) + j] = x[j]$.
    If the check succeeds, $\cA'$ simulates $\cA$ on $x[1:k(n)]$.
    If $x \in \cP^{(k)}_n$, then the first test always passes and $\cA$ accepts
    with probability at least $\frac 23$, so $\cA'$ accepts with probability
    at least $\frac 23$.
    Suppose $x$ is $\epsilon$-far from $\cP^{(k)}_n$.
    If $x[1:k(n)]$ is less than $\frac{\epsilon}{2}$-far from $\cP_{k(n)}$ and $x$ is less than $\frac{\epsilon}{2}$-far from being a repeated instance of $x[1:k(n)]$, then $x$ is less than $\epsilon$-far from $\cP_{k(n)}$.
    Thus, with probability at least $\frac 23$, at least one of the checks fails
    and $\cA'$ rejects.
    If $\cA$ makes $q(n)$ queries, then  $\cA'$ makes $q(k(n)) + O(1/\epsilon)$ queries, which is $O(q(k(n)))$
    since $\epsilon$ is a constant.

    Given an $\frac{\epsilon}{4}$-tester $\cT$ for $\cP_n^{(k)}$,
    we obtain an $\frac{\epsilon}{2}$-tester $\cT'$ for $\cP$ as follows.
    On input $x$ of sufficiently large length $n$, tester $\cT'$ first
    checks whether $n \in \image(k)$, and rejects if this is not the case.
    Since $\Lengths(\cP) \setminus \image(k)$ is finite, this only ignores finitely many nontrivial input lengths.
    Otherwise, $\cT'$ computes $n' = k^\dagger(n)$, which satisfies $k(n') = n$.
    Tester $\cT'$ then creates the instance $y = x^{\ceil{n'/n}}[1:n']$ and simulates
    $\cT$ on $y$. If $x \in \cP_n$, then $y \in \cP^{(k)}_{n'}$, and $\cT'$ accepts. If $x$
    is $\frac{\epsilon}{2}$-far from $\cP_n$, then $y$ is $\frac{\epsilon}{4}$-far from
    $\cP^{(k)}_{n'}$ (the factor of $2$ accounts for a possibly partial last repetition of $x$
    inside $y$), and $\cT'$ rejects. If $\cT$ makes $q'(n)$ queries, then $\cT'$ makes $q'(n')$ queries, implying $q'(n') \geq
    q(n) = q(k(n'))$ for infinitely many $n' \in \mathbb{N}$.

    \textbf{\cref{lem:rep-inst-2}.}
    We use the same construction of $\cA'$ as in the first part of \cref{lem:rep-inst-1}.
    Correctness thus follows immediately, and we only have to prove the runtime.
    Computing $k(n)$ takes time $O(t(k(n))$, and then for each of the
    $O(1/\epsilon)$ queries to check repetition, we incur a cost of $O(\log n)$
    for sampling random bits, a cost of $O(\log^{1+o(1)} n)$ for computing
    the index modulo $k(n)$ (note that $k(n) \le n$ here), and a cost of $O(\log
    n)$ to make the actual queries.
    If $\cA$ runs in $O(t(n))$ time, then its simulation takes $O(t(k(n)))$ time.
    In total, the running time of $\cA'$ is $O(t(k(n)) + \log^{1+o(1)} n)$

    \textbf{\cref{lem:rep-inst-3}.}
    We use the same construction of $\cT'$ as in the second part of \cref{lem:rep-inst-1}.
    Correctness thus follows immediately, and we only have to prove the runtime.
    By assumption, checking whether $n \in \image(k)$ and subsequently computing $n' = k^\dagger(n)$ takes $O(t(n) \log^{1+\gamma k^\dagger(n)})$ time.
    For each query made by $\cT$, tester $\cT'$ requires $O(\log^{1+o(1)} n')$ time to convert it to the corresponding query in $x$ via integer division.
    If $\cT$ runs in time $O(t(k(n)))$, then $\cT'$ runs in time 
    $O(t(k(n')) \log^{1+o(1)} n' + t(n) \log^{1+\gamma} k^\dagger(n))
    = O(t(n) \log^{1+\gamma} k^\dagger(n))$ time, as desired.
\end{proof}

\begin{remark}
    \label{rem:same-tester}
    For \textbf{Items 1} and \textbf{2} of
    \cref{lem:concat,lem:rep-inst}, if the same testers $\cA_1$, $\cA_2$, and $\cA$ achieve both the time
    and query upper bounds for $\cP^{(1)}$, $\cP^{(2)}$, and $\cP$, respectively, then the
    constructed testers $\cA$ and $\cA'$ also achieve the time and query upper bounds simultaneously.
\end{remark}

\begin{remark}
    \label{rem:lb-type-preserved-rep-inst}
    For \textbf{Item 3} of \cref{lem:rep-inst}, if the tester for $\cP^{(k)}$ succeeds only on infinitely many input lengths and $\image(k) \Delta \Lengths(\cP)$ is finite, then we get a tester for $\cP$ that succeeds on an infinite subsequence of $\Lengths(\cP)$.
\end{remark}

\subsubsection{Languages to properties}\label{sec:languages-to-properties}
To construct properties with a target time complexity, we first start with a language that is hard to \emph{decide}, and then map it through a suitable code so that \textsc{Yes} and \textsc{No} instances of the language are $\eps$-far, meaning \emph{testing} for this version of the language is as hard as \emph{deciding} the original hard language.
Below, we formalize this construction.


\begin{definition}[Properties from languages]
    \label{def:property-from-language}
    Let $L \subset \zo^*$ be a language. Define the property 
    $\cC = \cE(L)$ by
    $\cC \define \{ \cE(x) : x \in L \}$,
    where $\cE$ is the code from \cref{thm:code}.
\end{definition}


The following lemma translates the computational complexity of deciding language $L$ to that
of testing property $\cE(L)$.

\begin{lemma}
\label{lemma:property-from-language}
    Let $L \subset \zo^*$ be a language, and code $\cE$ and rate $r$ be as in \cref{thm:code}. Let $\cC \define \cE(L)$. Then for
    all sufficiently small constant $\epsilon > 0$, the following statements hold:
    \begin{enumerate}
        \item \label{lemma:property-from-language-1}
        If $L$ is decidable with error probability at most $\frac 13$ in
        $O(t(n))$ time, then there exists an $\epsilon$-tester for $\cC$ using
        $n$ queries and 
        $O(n \log n + t(\floor{r \cdot n}))$
        time.

        \item \label{lemma:property-from-language-2}
        If $\eps$-testing $\cC$ can be done in
        $O(t(\ceil{r \cdot n}))$
        time,
        then there exists a decider for $L$ with error probability at most $\frac 13$
        running in time $O(n \log n + t(n))$.
    \end{enumerate}
\end{lemma}
\begin{proof}
    \textbf{\cref{lemma:property-from-language-1}.} Let $\cA$ be an $O(t(n))$-time randomized algorithm
    for deciding $L$. We give a decider (\ie 0-tester) $\cT$ 
    which works as follows. On input $w$ of length $n$, tester $\cT$ immediately rejects if $n$ is not an integer multiple of $\frac 1 r$.
        Otherwise, $\cT$ runs the decoder from \cref{thm:code} on $w$. If $w$ is not a valid
        codeword, $\cT$ rejects. Otherwise, if there is a message $x$ (with $|x| =
        r \cdot n = \floor{r \cdot n}$) such that $\cE(x) = w$, the tester simulates
    the algorithm $\cA$ on input $x$ and accepts iff $\cA$ accepts.


    Correctness and query complexity follow directly from the construction.
    For the time complexity,
   checking the divisibility of $n$ by $\frac 1 r$ takes
    time $O(\log^{1+o(1)} n)$;
    querying all of $w$ and decoding it
    both take time $O(n \log n)$;
    and running algorithm $\cA$ takes time $O(t(|x|)) = O(t(\floor{r \cdot n}))$.

    \textbf{\cref{lemma:property-from-language-2}.} Let $\cT$ be an $\epsilon$-tester with running time
    $O(t'(n))$. We give a randomized algorithm $\cA$ for deciding $L$ as
    follows. On input $x$ of length $m$, the algorithm first computes $n \define \frac m r$ and runs the encoder from \cref{thm:code} to
    obtain string $w \define \cE(x)$ of length $n$. Then,
    $\cA$ simulates the tester $\cT$ on input $w$, and accepts iff
    $\cT$ accepts.

    We first show correctness. When $x \in L$, we have $w \in \cC_n$ by the
    construction of $\cC$, so $\cT$ accepts  with probability
    at least $\frac 23$, and so does $\cA$. When $x \not\in L$,
    the string $w$ is $\epsilon$-far from every other codeword for
        sufficiently small $\epsilon$ by \cref{thm:code}, so in particular
    $w$ is $\epsilon$-far from $\cC$.
    Hence $\cT$ rejects
    with probability at least $\frac 23$, and so does $\cA$.

    We now analyze the time complexity of $\cA$. Computing
    $w = \cE(x)$ takes time $O(m \log m)$ by \cref{thm:code}, and simulating $\cT$ takes times
    $O(t'(n)) = O(t'(m/r))$.
    If
   $t'(n) = O(t(\ceil{r \cdot n}))$,
    then this is
    at most $O(t(m))$. Hence the overall running time of $\cA$ is $O(m \log m +
    t(m))$.
\end{proof}

\begin{remark}
    \label{rem:lb-type-preserved-prop-from-lang}
    For \textbf{Item 2} of \cref{lemma:property-from-language}, if the tester for $\cC$ succeeds on infinitely many input lengths that are multiples of $\frac 1 r$, then we get a decider for $L$ that succeeds on infinitely many input lengths.
\end{remark}

\subsubsection{Combining the constructions}\label{sec:combining-contructions}
The next definition is used to produce problem instances with prescribed time complexity. Suppose a hard problem $\mathsf{P}$ has time complexity $T(k)$
for inputs of length $k$, with $T_\ell(k) \le T(k) \le T_u(k)$. Given a target function $t(n)$, we aim to construct, for each $n$, an instance of $\mathsf{P}$ of length $k = k(n)$ requiring time $T(k) \approx t(n)$ to solve. If $T$ were invertible, we could take $k(n) = T^{-1}(t(n))$, but since this may not hold, we use the following definition.

\begin{definition}[Attainment]
    \label{def:attains}
    Let $F, T_\ell, T_u, t : \bN \to \bR$ be nondecreasing and $k : \bN \to \bN$ be
    a function
    We say the pair $(T_\ell, T_u)$
    \emph{attains $t$ up to gap $F$ with length function $k$} if,
    for
    all large enough $n \in \bN$, we have
    $T_\ell(k(n)) \ge t(n)$ and $T_u(k(n)) \le F(t(n))$.
\end{definition}

The next definition gives the ``recipe'' by which, given a property $\cQ$ with nearly maximal query
complexity and a language $L$ that is hard to decide, as well as target query and time complexity
functions, we construct a property attaining the desired query and time complexities. We then
obtain our hierarchy theorems by plugging in specific choices for $\cQ$ and $L$.


\newcommand{\hp}[1]{\cref{def:hard-property}.\ref{#1}}

\begin{definition}[Hard property]
    \label{def:hard-property}
    Let $\cQ \subset \zo^*$ be an a.e.\ nonempty property satisfying the
    following conditions.
    \begin{enumerate}
        \item \label{hp-Q-complexity}
        $T_{\eps,n}(\cQ)=O(n \log n)$ for all constant $\eps\in(0,1)$ and $Q_{\eps,n}(\cQ)=\Omega\big(\frac n {\log n}\big)$ for some sufficiently small constant
        $\epsilon > 0$.
    
        \item \label{hp-Q-yes} $\cQ$ has an $O(n \log n)$-time
            YES query provider (\cref{def:query-providers}).
    \end{enumerate}
    Let $T_\ell, T_u : \bN \to \bR$ be $\Omega(n \log n)$ nondecreasing
    functions satisfying $T_\ell (n) \le T_u (n)$ for all large enough $n \in
    \bN$, and let $L \subset \zo^*$ be an a.e.\ nonempty language satisfying the
    following condition.
    \begin{enumerate}[resume]
        \item \label{hp-L-complexity}
        $L$ is decidable with error
        probability at most $\frac 13$ by an $O(T_u(n))$-time randomized algorithm, but no $O(T_\ell(n))$-time
         randomized algorithm.
    \end{enumerate}
    Let $q : \bN \to \bN$ and $t, F : \bN \to \bR$
    be nondecreasing functions and $k_* : \bN \to \bN$ be a function
    such that the following conditions hold.
    \begin{enumerate}[resume]
        \item \label{hp-qt}
        $q(n)$ is eventually surjective, and satisfies $q(n) \le n$ for all $n$
        and $q(n) \log q(n) = O\big(\frac{t(n)} {\log^{1.01} n}\big)$. Moreover,
        $q(n)$ is computable in $O(q(n) \log q(n))$ time.

        \item \label{hp-attainment}
        The pair $(T_\ell, T_u)$ attains $t$ up to gap $F : \bN \to \bR$ with
        eventually surjective
        length function $k_*$ (\cref{def:attains}).

        \item \label{hp-kprime-properties}
        $k_*(n) = O(q(n))$ and $k_*(n) \le r \cdot n$, where $r \in (0,1)$ is the
        constant from \cref{thm:code}. Moreover, $k_*(n)$ is computable in time
        $O(T_u(k_*(n)))$ and has a right inverse ${k_*}^\dagger$ such that
        ${k_*}^\dagger(n)$ is computable in time $O(T_\ell(n))$.
    \end{enumerate}
    %
    Let $\cC \define \cE(L)$ be the property from
    \cref{def:property-from-language},
    and define $k : \bN \to \bN$ by $k(n) \define \frac{k_*(n)}r$.
    Define the \emph{hard property} $\cP$ as follows:
    for each $n \in \bN$, let $\cP_{2n-1} \define \emptyset$ and
    $\cP_{2n} \define \concat\left( \cQ_n^{(q)}, \cC_n^{(k)} \right).$
\end{definition}

The following lemma bounds the complexity of testing the hard property
from \Cref{def:hard-property}.

\begin{lemma}[Complexity of testing the hard property]
    \label{lemma:complexity-hard-property}
    Let $\epsilon > 0$ be a sufficiently small constant. Let $\cQ$ be a property, $L$ be a
    language, and $q : \bN \to \bN$ and $t, F, T_\ell, T_u : \bN \to \bR$
    be functions satisfying \cref{def:hard-property} with theresulting property $\cP$. Then
    the following statements hold:
    \begin{enumerate}
        \item \label{lemma:complexity-hard-property-1}
        $\epsilon$-testing $\cP_{2n}$
        requires $\Omega\left(\frac{q(n)}{\log q(n)}\right)$ queries and $\Omega\left(\frac{t(n)}{\log^{2.01} (n)}\right)$ time.

        \item \label{lemma:complexity-hard-property-2}
        There exists an $\epsilon$-tester for
        $\cP_{2n}$
        using $O(q(n))$ queries and $O(q(n) \log^2 n + F(t(n)) \log n + \log^{2+o(1)} n)$ time.
    \end{enumerate}
\end{lemma}
\begin{proof}
    We prove the lower and upper bounds separately.

    \textbf{\cref{lemma:complexity-hard-property-1}.}
    First, we show the query complexity lower bound. By \cref{lem:rep-inst},
    using the assumption that $q$ is eventually surjective
    and the assumed query
    complexity lower bound for testing $\cQ_n$ (\hp{hp-Q-complexity}), 
    $Q_{2\eps,n}\big(\cQ^{(q)}\big)=\Omega\big(\frac{q(n)}{ \log(q(n))}\big)$.
    Also, $\cQ^{(q)}$ is a.e.\ nonempty by \cref{def:repeated-instances}
    and the assumption that $\cQ$ is a.e.\ nonempty.
    Since $L$ is a.e.\ nonempty,
    so is $\cC^{(k)}$
    by \cref{def:repeated-instances,def:property-from-language}
    and because $k(n)$ is an integer multiple of $\frac 1 r$.
    Thus,
    $Q_{\eps,2n}(\cP)=\Omega\big(\frac{q(n)}{ \log(q(n))}\big)$ by
    \cref{lem:concat}.\ref{lem:concat-1}.


    Next, we show the time complexity lower bound. Since by assumption (\hp{hp-L-complexity}), $L$
    cannot be decided by any randomized $O(T_\ell(n))$-time algorithm and
    $T_\ell(n) = \Omega(n \log n)$,
    \cref{lemma:property-from-language}.\ref{lemma:property-from-language-2} implies that
    there is no
    $O(T_\ell(\ceil{r \cdot n}))$-time
    tester for $\cC$.

    We claim this implies, via \cref{lem:rep-inst}.\ref{lem:rep-inst-3},
    that there is no $O(T_\ell(r \cdot k(n)) / \log^{1.005} n)$-time tester for
    $\cC_n^{(k)}$. First, note that $k(n) \le n$ since $k(n) = k_*(n)/r$ and
    $k_*(n) \le r \cdot n$ (\hp{hp-kprime-properties}).
    Also, $\Lengths(\cC_n) \setminus \image(k)$
    is finite because $\Lengths(\cC_n)$ only contains integer multiples
    of $1/r$ by \cref{def:repeated-instances,def:property-from-language},
    while $\image(k)$ contains a.e.\ integer multiple of $1/r$ because
    $k_*(n)$ is eventually surjective.
    Since ${k_*}^\dagger(n)$ is computable in time $O(T_\ell(n))$ by
    \hp{hp-kprime-properties},
    a right inverse $k^\dagger(n)$ of $k$ is computable (by first performing an
    integer division by $1/r$) in time
    $O(\log^{1+o(1)}(n) + T_\ell(\floor{r \cdot n})) = O(T_\ell(\floor{r \cdot n}))$,
    and similarly $\image(k)$ is also decidable in this time. Hence applying
    \cref{lem:rep-inst}.\ref{lem:rep-inst-3} with parameter $\gamma=0.005$ and
    function $t'(m) = T_\ell(\floor{r \cdot m}) / \log^{1.005} k^{-1}(m)$ yields
    that there is no $O(T_\ell(r \cdot k(n))/\log^{1.005} n)$-time
    tester for $\cC_{n}^{(k)}$.

    Since $T_\ell(r \cdot k(n)) = T_\ell(k_*(n)) \ge t(n)$ by attainment
    (\cref{def:attains}), we conclude that there is no
    $O(t(n) / \log^{1.005} n)$-time tester for $\cC_n^{(k)}$.
    By assumption (\hp{hp-Q-yes}), $\cQ$ has an $O(n \log n)$-time
    YES query provider, which we can use together with the
    computability of $q(n)$ in $O(q(n) \log q(n))$ time to derive an
    $O(q(n) \log q(n))$-time YES query provider for
    $\cQ_n^{(q)}$. Because $O(q(n) \log q(n)) = O(t(n) / \log^{1.01} n)$ by
    assumption (\hp{hp-qt}), an $O(t(n) / \log^{2.01}
    n)$-time tester for $\cP_{2n}$ would yield by
    \cref{lem:concat}.\ref{lem:concat-3} an
    $O(t(n) / \log^{1.005} n)$-time tester for $\cC_n^{(k)}$,
    which is a contradiction.

    \textbf{\cref{lemma:complexity-hard-property-2}.}
    First, we show the query complexity upper bound. Since $\cQ_n$ is
    trivially testable using $n$
    queries, \cref{lem:rep-inst}.\ref{lem:rep-inst-1} implies that
    $Q_{\epsilon,n}(\cQ^{(q)}) = O(q(n))$.
    By \cref{lemma:property-from-language},
    $Q_{\epsilon/2,n}(\cC) = O(n)$,
    so
    applying \cref{lem:rep-inst}.\ref{lem:rep-inst-1},
    $Q_{\epsilon,n}(\cC^{(k)}) = O(k(n))$,
    which is $O(q(n))$ by \hp{hp-kprime-properties}.
    Applying \cref{lem:concat}.\ref{lem:concat-1}
    to $\cQ_n^{(q)}$ and $\cC_n^{(k)}$, we obtain
    that
    $Q_{\epsilon,2n}(\cP) = O(q(n))$.

    Next, we show the time complexity upper bound.
    By assumption (Definitions~\ref{def:hard-property}.\ref{hp-Q-complexity}
    and \ref{def:hard-property}.\ref{hp-qt}),
    $T_{\epsilon/2,n}(\cQ) = O(n \log n)$
    and $q(n)$ is computable in time $O(q(n) \log q(n))$, so applying
    \cref{lem:rep-inst}.\ref{lem:rep-inst-2},
    $T_{\epsilon,n}(\cQ^{(q)}) = O(q(n) \log q(n) + \log^{1+o(1)} n)$.

    By \hp{hp-L-complexity} and
    \cref{lemma:property-from-language}.\ref{lemma:property-from-language-1},
    $T_{\epsilon/2,n}(\cC) = O(n \log n + T_u(\floor{r \cdot n}))
    = O(T_u(\floor{r \cdot n}))$, recalling that $T_u(m) = \Omega(m \log m)$.
    By \hp{hp-kprime-properties}, $k_*(n)$ is computable in time
    $O(T_u(k_*(n)))$, and hence $k(n) = k_*(n)/r$ is computable in time
    $O(T_u(k_*(n)) + \log^{1+o(1)}(k_*(n))) = O(T_u(r \cdot k(n)))$.
    Hence \cref{lem:rep-inst}.\ref{lem:rep-inst-2} implies that
    $T_{\epsilon,n}(\cC^{(k)}) = O(T_u(r \cdot k(n)) + \log^{1+o(1)} n)
    = O(T_u(k_*(n)) + \log^{1+o(1)} n) = O(F(t(n)) + \log^{1+o(1)} n)$,
    where the last step is valid by attainment (\cref{def:attains}).

    Applying \cref{lem:concat}.\ref{lem:concat-2} to $\cQ_n^{(q)}$ and $\cC_n^{(k)}$,
    we get that
    $T_{\epsilon,2n}(\cP) = O(q(n) \log q(n) \log n +
    F(t(n)) \log n + \log^{2+o(1)} n) = O(q(n) \log^2 n + F(t(n)) \log n + \log^{2+o(1)} n)$.

    By \cref{rem:same-tester}, the same tester achieves both the query and time complexity upper bound.
\end{proof}

\begin{remark}
    \label{rem:attain-relax}
    \cref{lemma:complexity-hard-property}.\ref{lemma:complexity-hard-property-1}
    also admits the following \textbf{infinitely often} version. If we
    \emph{strengthen} the hypothesis of \cref{def:hard-property} by requiring that the
    $T_\ell(n)$ lower bound for deciding language $L$ holds even against
    algorithms that succeed on infinitely many input lengths, and then
    \emph{weaken} the notion of attainment in \cref{def:attains} by only
    requiring the inequality $T_\ell(k(n)) \ge t(n)$ on infinitely many $n \in
    \bN$, then a straightforward modification of the proof of
    \cref{lemma:complexity-hard-property}.\ref{lemma:complexity-hard-property-1}
    (using
    \cref{rem:lb-type-preserved-concat,rem:lb-type-preserved-rep-inst,rem:lb-type-preserved-prop-from-lang})
    yields that any tester for $\cP_{2n}$ that succeeds in the usual sense (\ie
    at a.e.\ input length) must use $\Omega\left(\frac{t(n)}{\log^{2.01}
    n}\right)$ time (on infinitely many input lengths).
    
\end{remark}

\subsection{Hierarchy theorems}\label{sec:hierarchies-proof}
\subsubsection{Unconditional weak hierarchy theorem (proof of \cref{thm:weak})}
\label{sec:hierarchies-proof-weak}

In this subsection, we combine the property $\cQ$ from \cref{sec:3cnf-property} with the hard
language $L$ obtained in \cref{sec:hard-lang-ram} to satisfy the requirements of
\cref{def:hard-property}, and then use \cref{lemma:complexity-hard-property} to establish
\cref{thm:weak}.

We start with the preconditions of \cref{thm:weak}.
Specifically, we are given a constant $\epsilon \in (0, 1)$, and nondecreasing, unbounded functions $q, t : \bN \to \bN$
which satisfy the following conditions.
\begin{enumerate}
    \item The function $n \mapsto q(2n)$ is eventually surjective,
        $q(n)$ is computable in time $O(q(n) \log q(n))$, and
    $q^{-1}(n)$ is computable in time $O(t \circ q^{-1} (n))$.
    \item $t(n)$ is computable in time $O(t(n))$ and
    $t^{-1}(n)$ is computable in time $O(n)$.
    \item $q(n) \leq \min\left\{ \frac{r \cdot n}{2}, \frac{t(n)}{\log^{2.01} n} \right\}$,
        where $r$ is the constant from \cref{thm:code}.
\end{enumerate}

We make the following definitions.
\begin{enumerate}
    \item Let $\cQ$ be the property from \cref{def:3cnf-construction}.
    
    \item Let $q_* : \bN \to \bN$ and $t_* : \bN \to \bR$
            be given by $q_*(n) \define q(2n)$ and
        $t_*(n) \define t(2n) \log^{2.01} n$.
    
    \item Let $T_\ell, T_u : \bN \to \bR$ be given by
    $T_\ell (n) \define t_* \circ q_*^{-1} (n)
    \text{ and } T_u(n) \define 2^{T_\ell (n)^{2.01}}$.
    
    \item Let $L$ be the language from \cref{thm:hard-lang-ram} with
    parameter\footnote{Formally, we invoke \cref{thm:hard-lang-ram},
    but with a slightly better exponent in the upper bound (say, $2.009$) and a
    function $t_*'(n) \ge t_*(n)$ obtained by approximating $\log^{2.01} n$ by
    an integer (a constant-factor approximation suffices). By assumption,
    computing $q^{-1}(n)$ and hence $q_*^{-1}(n)$ takes time $O(t \circ q^{-1}
    (n)) = O(t_* \circ q_*^{-1}(n))$, and subsequently computing $t_*' \circ
    q_*^{-1} (n)$ takes time $O(t_*' \circ q_*^{-1} (n))$, meaning $T_\ell$ is
    efficiently constructible and the call to \cref{thm:hard-lang-ram} is
    legal.} $T_\ell$.
    
    \item Let $k_* : \bN \to \bN$ be given by
        $k_*(n) \define q_*(n)$.

    \item Let $F : \bN \to \bR$ be given by $F(n) \define 2^{n^{2.01}}$.
\end{enumerate}

\begin{lemma}
    \label{lemma:weak-choices}
    The choices above of
    $\cQ, T_\ell, T_u, L, q_*, t_*, k_*, F$
    satisfy the conditions of \cref{def:hard-property}.
\end{lemma}
\begin{proof}
    {\bf Conditions on $\cQ$.}
    The first condition, that $\cQ$ be a.e.\ nonempty, holds by the
    construction in \cref{def:3cnf-construction}. \cref{def:hard-property} also requires
    $\epsilon$-testing property $\cQ_n$ to have query complexity $\Omega(\frac n  {\log n})$ and time
    complexity $O(n \log n)$, which is guaranteed by
    \cref{lemma:3cnf-query-bound,lemma:3cnf-time-bound}, respectively. Finally, $\cQ$
    must have an
    $O(n \log n)$-time YES query provider. This is easy to ensure for any reasonable scheme for
    encoding 3CNF formulas into binary strings (which we leave implicit for brevity); for example,
    we can use an empty formula with a trivial all-$0$s assignment, and the encoder from
    \cref{thm:code}.

    {\bf Conditions on $T_\ell, T_u, L$.}
    The first condition is for $T_\ell (n)$ and $T_u (n)$ to be nondecreasing and $\Omega(n \log n)$, with $T_\ell \leq T_u$.
    By assumption, $t(n) \geq q(n) \log^2 n$ and $q(n) \leq n$.
    This yields
    \[
            t_* \circ q_*^{-1}(n)
            \ge t(2 \cdot q^{-1}(n)/2)
            \ge q(q^{-1}(n)) \log^2 n
            \ge n \log^2 n \,,
    \]
    giving us that $T_\ell (n) = \Omega(n \log n)$.
    \cref{def:hard-property} also requires $L$ to be a.e.\ nonempty; decidable by an $O(T_u(n))$-time randomized algorithm but not by any $O(T_\ell(n))$-time randomized algorithm, which are all true by \cref{thm:hard-lang-ram}.

    {\bf Conditions on $q_*, t_*$.}
    The conditions that
    $q_*, t_*$ be nondecreasing, that $q_*(n)$ be eventually surjective and satisfy
        $q_*(n) \le n$ and $q_*(n) \log q_*(n) = O(t_*(n) / \log^{1.01} n)$,
        and that $q_*(n)$ be computable in $O(q_*(n) \log q_*(n))$ time
        all hold by the assumptions on $q, t$ and the construction of $q_*, t_*$.

    {\bf Attainment.}
    \cref{def:hard-property} requires the pair $(T_\ell, T_u)$ to attain $t_*$
    up to gap $F$ with length function $k_*$.
    By construction, $T_\ell(k_*(n)) = t_* \circ q_*^{-1}(q_*(n)) = t_* (n)$
    for infinitely many $n \in \bN$.
    Since the construction from \cref{thm:hard-lang-ram} rules out {\bf infinitely often} decidability, then by \cref{rem:attain-relax}, relaxing the definition of attainment to infinitely many input lengths is still valid.
    
    For the upper bound, we have
    \[
        T_u(k_*(n)) = 2^{T_\ell(k_*(n))^{2.01}} = F(T_\ell(k_*(n))) = F(t_* \circ q_*^{-1} \circ q_* (n)) \leq F(t_*(n)).
    \]
    
    {\bf Conditions on $k_*$.}
    The conditions that $k_*(n) = O(q_*(n))$ and $k_*(n) \leq r \cdot n$ follow from construction and the assumptions on $q$.
    For computability of $k_*(n)$, we use the assumption on the computability of $q(n)$ and the lower bound on $T_u(n)$ to get that $k_*(n)$ is computable in time $O(k_*(n) \log k_*(n)) = O(T_u (k_*(n)))$.
    Similarly, for the computability of $k_*^\dagger(n)$, we use the assumption on the computability of $q^{-1} (n)$ and the lower bound on $T_\ell(n)$ to get that
    $k_*^\dagger(n)$ is computable in time $O(t \circ q^{-1} (n) + \log^{1 + o(1)} q^{-1} (n)) = O(t_* \circ q_*^{-1} (n)) = O(T_\ell (n))$.
\end{proof}

Combining \cref{lemma:weak-choices,lemma:complexity-hard-property} yields the following corollary,
which implies \cref{thm:weak}.

\begin{corollary}
    \label{cor:weak-hierarchy-const}
    \sloppy
    Let $\cP$ be the hard property obtained by applying \cref{def:hard-property} with the choices above.
    Then $\epsilon$-testing $\cP_{2n}$ requires $\Omega\left(\frac{q_*(n)}{\log q_*(n)}\right) = \Omega\left(\frac{q(2n)}{\log q(2n)}\right)$ queries and $\Omega \left( \frac{t_*(n)}{\log^{2.01} n} \right) = \Omega \left( t(2n) \right)$ time.
    Moreover, there exists an $\epsilon$-tester for $\cP_{2n}$ using $O(q_*(n)) = O(q(2n))$ queries and $O(q_*(n) \log^2 n + F(t_*(n)) \log n + \log^{2+o(1)} n) = O(2^{t(2n)^{2.01} \log^{4.05} (2n)})$ time.
\end{corollary}


\subsubsection{Strong hierarchy theorem from SETH (proof of \cref{thm:strong})}
\label{sec:hierarchies-proof-strong}

In this subsection, we combine the property $\cQ$ from \cref{sec:3cnf-property} with the hard
language $L$ obtained from SETH in \cref{sec:seth-language} to satisfy the requirements of
\cref{def:hard-property}, and then use \cref{lemma:complexity-hard-property} to establish
\cref{thm:strong}.
Toward this goal,  assume the hypotheses of \cref{thm:strong}. Specifically, let $\epsilon,
\gamma > 0$ be sufficiently small constants and $q, t : \bN \to \bN$ be nondecreasing unbounded
functions satisfying the following conditions.
\begin{enumerate}
    \item The function $n \mapsto q(2n)$ is eventually surjective and
        $q(n)$ is computable in time $O(q(n) \log q(n))$.
    \item $t(n)$ is computable in time $O(t(n))$, and
    $t^{-1}(n)$ is computable in time $O(n)$.
    \item $q(n) \leq \min\left\{ \frac{r \cdot n}{2}, \frac{t(n)}{\log^{2.01} n} \right\}$,
        where $r$ is the constant from \cref{thm:code}.
    \item 
        $\log t(n) \log\log t(n)$ has slope $o(1)$
        and $t(n) \le 2^{O\left(\frac{q(n)}{\log q(n)}\right)}$.
\end{enumerate}

We then make the following definitions.

\begin{enumerate}
    \item Let $\cQ$ be the property from \cref{def:3cnf-construction}.
    \item Let $L$ be the language $L^{(\gamma)}$ from \cref{def:hard-language}.
    \item Let $T_\ell, T_u : \bN \to \bR$ be given by 
        $T_\ell(n) \define 2^{\frac{(1-\gamma)n}{A_\gamma \log n}}$
        and $T_u(n) \define 2^{\frac{n}{A_\gamma \log n}}$, where $A_\gamma >
        0$ is the constant from \cref{lemma:bounds-hard-language}.
    \item Let $q_*, t_* : \bN \to \bN$
            be given by $q_*(n) \define q(2n)$ and $t_*(n) \define t(2n)$.
    \item Define the auxiliary function $f : \bN \to \bR$ by $f(m)
        \define \frac{A_\gamma}{1-3\gamma} \log m \log\log m$.
        Then $k_* : \bN \to \bN$ is an explicit computable function,
        specified in \cref{lemma:strong-choices} below, and
        $\abs*{k_*(n) - \floor{f(t_*(n))}} \le 1$.
    \item Let $F : \bN \to \bR$ be given by $F(n) \define n^{\frac{1}{1-4\gamma}}$.
\end{enumerate}

\begin{lemma}
    \label{lemma:strong-choices}
    The choices above of $\cQ, L, T_\ell, T_u, k_*, q_*, t_*, F$ satisfy the conditions of
    \cref{def:hard-property}.
\end{lemma}
\begin{proof}
    \textbf{Conditions on $\cQ$}  have already been established in \cref{lemma:weak-choices}.

    \textbf{Conditions on $L$.}
    The first condition is for $T_\ell(n)$ and $T_u(n)$ to be nondecreasing and $\Omega(n \log n)$,
    with $T_\ell \le T_u$; this is true by construction. \cref{def:hard-property} also requires the
    language $L$ to be a.e.\ nonempty, which we satisfy
    by choosing an appropriate encoding scheme for formulas into binary strings in
    \cref{def:hard-language}. Finally, $L$ must be decidable by an $O(T_u(n))$-time randomized
    algorithm but not by any $O(T_\ell(n))$-time randomized algorithm, which is true by
    \cref{lemma:bounds-hard-language}.

    \textbf{Conditions on $q_*,t_*$.}
    The conditions that $q_*,t_*$ be nondecreasing, that $q_*$ be eventually
    surjective and satisfy $q_*(n) \le n$
    and $q_*(n) \log q_*(n) = O(t_*(n) / \log^{1.01} n)$, and that $q_*(n)$ be
    computable in $O(q_*(n) \log q_*(n))$ time hold by the assumptions on $q,t$ and
    the definition of $q_*,t_*$.

    \textbf{Attainment.}
    \cref{def:hard-property} requires that the pair $(T_\ell, T_u)$ attain $t_*$
    up to gap $F$ with length function $k_*$.
    We first show that $T_\ell(k_*(n)) \ge t_*(n)$ for all sufficiently large $n$.
    Since $t_*$ is unbounded, for all sufficiently large $n$, we have
    \[
        k_*(n)
        \ge \floor{\frac{A_\gamma}{1-3\gamma} \log t_*(n) \log\log t_*(n)} - 1
        \ge \frac{(1+0.1\gamma) A_\gamma}{1-\gamma} \log t_*(n) \log\log t_*(n)
    \]
    and hence, as desired,
    \[
        T_\ell(k_*(n))
        = 2^{\frac{(1-\gamma) k_*(n)}{A_\gamma \log k_*(n)}}
        \ge 2^{\frac{(1+0.1\gamma) \log t_*(n) \log\log t_*(n)}
            {\log\left(\frac{A_\gamma}{1-3\gamma} \cdot \log t_*(n) \log\log t_*(n)\right)}}
        \ge 2^{\frac{(1+0.1\gamma) \log t_*(n) \log\log t_*(n)}
            {\log\left(\left(\log t_*(n)\right)^{1+0.1\gamma}\right)}}
        = 2^{\log t_*(n)}
        = t_*(n) \,.
    \]
    We now show that $T_u(k_*(n)) \le F(t_*(n))$ for all sufficiently large $n$.
    As above, for all sufficiently large $n$, we have
    \[
        k_*(n)
        \le \floor{\frac{A_\gamma}{1-3\gamma} \log t_*(n) \log\log t_*(n)} + 1
        \le \frac{A_\gamma}{1-4\gamma} \log t_*(n) \log\log t_*(n)
    \]
    and hence, as desired,
    \[
        T_u(k_*(n))
        = 2^{\frac{k_*(n)}{A_\gamma \log k_*(n)}}
        \le 2^{\frac{\log t_*(n) \log\log t_*(n)}{(1-4\gamma)
            \log\left(\frac{A_\gamma}{1-3\gamma} \log t_*(n) \log\log t_*(n)\right)}}
        \le \left(2^{\frac{\log t_*(n) \log\log t_*(n)}{\log\log t_*(n)}}\right)^{\frac{1}{1-4\gamma}}
        = F(t_*(n)) \,.
    \]

    \textbf{Conditions on $k_*$.}
    We now complete the definition of $k_*$ and
    show that it is eventually surjective with $k_*(n) =
    O(q_*(n))$ and $k_*(n) \le r \cdot n$, and moreover that $k_*(n)$ is
    computable in time $O(T_u(k_*(n)))$ and ${k_*}^\dagger(n)$ is computable in time
    $O(T_\ell(n))$.

    We use the following fact from numerical analysis, see \eg \cite{Brent76}: given a number
    $x$ with $\ell$ bits of precision, it is possible to compute $\log x$ up to $\ell$ bits of
    precision in time $O(\ell \polylog \ell)$. It follows that the function $f(m) =
    \frac{A_\gamma}{1-3\gamma} \log m \log\log m$ may be computed up to arbitrarily small constant
    additive error in time $O(\log m \poly\log\log m)$.

    We first complete the definition of $k_*(n)$ by specifying the algorithm computing it
    in time $O(T_u(k_*(n)))$. Note that the desired upper bound is
    \[
        O(T_u(k_*(n))) \ge O(T_\ell(k_*(n))) \ge O(t_*(n)) \,,
    \]
    where the last step is valid by attainment. Recall that $k_*(n)$ must satisfy
    $\abs*{k_*(n) - \floor{f(t_*(n))}} \le 1$. Now, $t_*(n)$ is
    computable in time $O(t_*(n))$ by the assumption on $t$,
    and given $t_*(n)$, using the observation above we may estimate
    $f(t_*(n))$ up to additive error at most $\frac 1 8$ -- call this
    estimate $\widetilde f$ -- in time $O(\log t_*(n) \poly\log\log
    t_*(n)) = O(t_*(n))$. We then define $k'(n) \define \floor{\widetilde f}$,
    and observe that it satisfies $\abs*{k'(n) - f(t_*(n))} \le 1$.
    Thus the computability claim holds.

    Second, we claim that $k_*$ is eventually surjective with
    $k'(n) = O(q_*(n))$ and $k_*(n) \le r \cdot n$. The first claim holds
    because the function $n \mapsto f(t_*(n))$ has slope $o(1)$ by
    the hypothesis that $\log t(n) \log\log t(n)$ has slope $o(1)$,
    along with the definition of $k_*(n)$ above. Since $t(n) \le
    2^{O\left(\frac{q(n)}{\log q(n)}\right)}$, which implies that $t_*(n) \le
    2^{O\left(\frac{q_*(n)}{\log q_*(n)}\right)}$, we get $k_*(n) =
    O\left(\log t_*(n) \log\log t_*(n)\right) = O(q_*(n))$, as desired. Finally,
    since $\log t_*(n) \log\log t_*(n)$ has slope $o(1)$, we get $k_*(n) =
    O\left(\log t_*(n) \log\log t_*(n)\right) = o(n)$, so $k_*(n) \le r \cdot n$
    for all sufficiently large $n$.

    Finally, we claim that ${k_*}^\dagger(n)$ is computable in
    $O(T_\ell(n))$ time. That is, there is an algorithm which,
    on input $n$, outputs some $n' \in \bN$ satisfying $k_*(n') = n$ in
    the announced time. The desired upper bound is
    \[
        O\left(T_\ell(n)\right)
        = O\left(T_\ell(k_*(n'))\right)
        \ge O(t_*(n')) \,,
    \]
    where the last step is valid by attainment. The algorithm works as follows: on input $n$, it
    first performs exponential search to find an integer $m$ satisfying
    $f(m) \in \left[ n + \frac{1}{4}, n + \frac{3}{4} \right],$
    which exists because $s \mapsto f(t_*(s))$
    has slope $o(1)$. Specifically, there exists
    $n' \in \bN$ satisfying $\abs*{f(t_*(n')) - \left(n + \frac 1 2\right)} \le \frac 1 {16}$, so by
    estimating $f(m_i)$ up to additive error $\frac 1 {16}$ in  time $O(\log m_i \poly\log\log m_i)$ in
    each step $i$ of the exponential search and stopping when this estimate is within distance $\frac 1 8$
    of $n + \frac 1 2$, we can find integer $m$ satisfying $\abs*{f(m) - \left(n + \frac 1 2\right)}
    \le \frac 1 4$, as desired. Since $m$ satisfies $m \le O(2^n)$ and $n \le O(\log^2 m)$,
    the exponential search takes  $O(n)$ steps, and hence $m$ can be found
    in time $O(\polylog m)$.

    Then, the algorithm uses the invertibility of $t$ to compute in $O(m)$ time
    the number $n'$ satisfying $t_*(n'-1) \le m \le t_*(n')$,
    and outputs $n'$. We claim that $k_*(n') = n$. Note that
    $f(t_*(n'-1)) \le f(m) \le n + \frac 3 4$, which implies that
    $f(t_*(n')) \le n + \frac 3 4 + \frac 1 {16} = n + \frac {13}{16}$ since $s
    \mapsto f(t_*(s))$ has slope $o(1)$. It follows that $k_*(n') =
    \floor{\widetilde f} \le \floor{n + \frac {13}{16} + \frac 1 8} = n$. Similarly, $f(t_*(n')) \ge f(m) \ge n + \frac 1 4$, and hence $k_*(n')
    = \floor{\widetilde f} \ge \floor{n + \frac 1 4 - \frac 1 8} = n$. Thus $k_*(n')
    = n$, as claimed. The total running time is $O(\polylog m) + O(m) = O(m)
    \le O(t_*(n'))$, as needed.
\end{proof}

Combining \cref{lemma:strong-choices,lemma:complexity-hard-property} yields the following corollary,
which implies \cref{thm:strong}.

\begin{corollary}
    Let $\cP$ be the hard property obtained by applying \cref{def:hard-property}
    with the choices above. Then under \cref{assumption:seth},
    $\epsilon$-testing $\cP_{2n}$ requires
    $\Omega\left(\frac{q_*(n)}{\log q_*(n)}\right) =
    \Omega\left(\frac{q(2n)}{\log q(2n)}\right)$ queries and
    $\Omega\left(\frac{t_*(n)}{\log^{2.01} n}\right) =
    \Omega\left(\frac{t(2n)}{\log^{2.01} (2n)}\right)$ time. Moreover, there
    exists an $\epsilon$-tester for $\cP_{2n}$ using $O(q_*(n))
    = O(q(2n))$ queries and $O\left(t_*(n)^{\frac{1}{1-4\gamma}} \log
    n\right) = O\left(t(2n)^{\frac{1}{1-4\gamma}} \log (2n)\right)$ time.
\end{corollary}

\section{Distribution-free distance approximation for halfspaces}\label{sec:distribution-free-lb-halfspaces}

In this section, we prove a fine-grained hardness result for distribution-free
distance approximation to halfspaces in the low-dimensional setting. We first
formally define distribution-free distance approximation and the halfspace property over $\bZ^d$.

\begin{definition}[Distribution-free distance approximation]
    \label{def:distribution-free-distance-approximation}
    Let $\cX$ be a universe set and $\cP \subset \zo^\cX$ be a property of
    Boolean functions over $\cX$. 
    Define the \emph{distance from a function $f:\cX\to \zo$
    to property $\cP$ with respect to a distribution $\cD$ over $\cX$} as
    \[
        \dist_\cD(f, \cP)
        \define \inf_{g \in \cP} \dist_\cD(f, g) \,,
        \quad \text{where} \quad
        \dist_\cD(f, g) \define \Pru{x \sim \cD}{f(x) \ne g(x)} \,.
    \]
    A randomized algorithm $\cA$ is a \emph{distribution-free distance approximation}
    algorithm for property $\cP$ if,
    for each probability distribution $\cD$ over $\cX$
    and function $f : \cX \to \zo$,
    given input parameters $\epsilon,\delta\in (0,1)$,
    algorithm $\cA$ uses \emph{labeled samples} of the form $(x \sim
    \cD, f(x))$ and queries to $f$ to output a number $\widehat \alpha$
    which, with probability at least $1-\delta$, satisfies $\abs*{\dist_\cD(f, \cP) -
    \widehat \alpha} \le \epsilon$. 
     By default, $\delta=\frac 13$.
\end{definition}

\begin{definition}[Hyperplanes and halfspaces]
    \label{def:halfspaces}
    A $(d-1)$-dimensional \emph{hyperplane} $H$ in $\bR^d$ is the set of points $x \in
    \bR^d$ satisfying $\inp{w}{x} + \theta = 0$, where $w \in \bR^d$ and $\theta
    \in \bR$ are the parameters defining~$H$.
    A function
    $f : \bZ^d \to \zo$ is a \emph{halfspace} if there exists a hyperplane $H$
    with parameters $w, \theta$
    such that $f(x) = \ind{\inp{w}{x} + \theta \ge 0}$
    for all $x \in \bZ^d$. We denote the class of
    halfspaces over $\bZ^d$ by $\cH$.
\end{definition}

Our result is the following.

\begin{theorem}
    \label{thm:distribution-free-testing}
    Under the integer $k$-SUM conjecture, 
    for all
    constants
    $d \in \bN$ and $\gamma > 0$, there is no distribution-free
    distance approximation algorithm for halfspaces over $\bZ^d$
    running in time $(1/\epsilon)^{\ceil{(d+1)/2} - \gamma}$ (as a function of
    $\epsilon$). This lower bounds holds\footnote{The significance of stating the lower bound for distributions supported on
    points with bounded coordinates is that, in our RAM model with logarithmic
    cost, one could obtain trivial lower bounds by constructing inputs that
    are arbitrarily expensive to read. The stated bound implies only a
    polylogarithmic overhead for reading the input, which keeps the
    bound meaningful.} even if the input
    distribution $\cD$ is promised to be supported on points with absolute
    coordinate values at most $(1/\epsilon)^{O_d(1)}$.
\end{theorem}

We use the same RAM model as in \cref{sec:hierarchies}, with modifications to match the setting of
\cref{thm:distribution-free-testing}. The input tape is intially empty. The dimension $d \in \bN$ is
part of the problem description, and the error parameter $\epsilon$ is given on the parameter tape,
encoded in $O(\log(1/\epsilon))$ bits.\footnote{If $\bin{\eps}$ is too long, we can reduce
$\epsilon$ by at most a factor of 2 to get the desired length without affecting asymptotics.} To
model access to a distribution $\cD$ and a function $f : \bZ^d \to \zo$, we add two unit
time operations:

\textsc{DrawSamp}:  Clears the input tape and resets its head, then samples $(x, f(x))$ with $x \sim
\cD$, and writes $x_1, \dots, x_d, f(x)$ on the input tape (recall that $x \in \bZ^d$ and $f(x)$ is
a bit).

\textsc{QueryFunc}$(X_j)$: Clears the input tape and resets its head, sets $x \define (X_j, \dots,
X_{j+d-1})$ using $d$ machine registers (see \cref{table:ram}), and writes $f(x)$ on the input tape.

We prove \cref{thm:distribution-free-testing} via a reduction from the integer
$k$-SUM problem, stated next.

\begin{problem}[$k$-SUM]
    \label{problem:k-sum}
    Given $k$ lists $A_1, \dotsc, A_k$ of $n$ distinct integers each, where each
    integer lies in the range $[-n^{2k}, n^{2k}]$, is there set of choices $a_i
    \in A_i$ for each $i \in [k]$ such that $a_1 + \dotsm + a_{k-1} = a_k$?
\end{problem}

Our hardness result is based on the following standard conjecture on the time
complexity of the $k$-SUM problem (see the lecture
notes \cite{vassilevska2020lecture9} for further background on this
conjecture).

\begin{conjecture}[$k$-SUM conjecture \cite{AbboudL13}]
    \label{conjecture:ksum}
    For each constant $k \ge 2$ and $\gamma > 0$, no randomized algorithm can solve
    \cref{problem:k-sum}  in $O(n^{\ceil{k/2}-\gamma})$  time with error probability at most $\frac
    13$.
\end{conjecture}

\begin{remark}
    To align with our RAM model in \cref{sec:hierarchies} and avoid complications with real-number
    representation, we study distance approximation over $\bZ^d$ rather than $\bR^d$. Our arguments
    extend to the real setting under a suitable real RAM model and the $k$-SUM conjecture; see
    \cref{remark:real-values}.
\end{remark}

\subsection{Proof of \cref{thm:distribution-free-testing}}

Distance approximation for halfspaces is closely related to geometric decision
problems 
known to be $k$-SUM hard. For example, the \textsc{Geombase}
problem\footnote{Given $n$ points $(x_1, y_1), \dotsc, (x_n, y_n)$ on the plane,
with $y_i \in \{0,1,2\}$ for all $i \in [n]$, do any three of these points lie
on a non-horizontal line?} is a classic 3-SUM hard problem \cite{GajentaanO95},
and deciding whether $d+1$ points in $\bR^d$ lie on a hyperplane is $(d+1)$-SUM
hard \cite{Erickson96}. We build upon this theme by reducing from $(d+1)$-SUM to
distance approximation for halfspaces.

The starting point for our proof is
a construction from \cite{FinkHKS17}
for the following high-dimensional version of the \textsc{Geombase} problem.
%
A $(d-1)$-dimensional hyperplane in $\bR^d$ is \emph{vertical} if it contains a
vertical line, \ie two distinct points $p, q \in \bR^d$ satisfying $p_i = q_i$
for all $i \in [d-1]$; otherwise, it is \emph{non-vertical}.
The problem is: given $n$ points in $\bR^d$, do any $d+1$ of
them lie on a non-vertical $(d-1)$-dimensional hyperplane? This problem is $(d+1)$-SUM hard
and was used by \cite{FinkHKS17} to prove
lower bounds for
linear separability of probabilistic
point sets.
We state the $(d+1)$-SUM hardness result in the real-valued setting, which is the
norm in computational geometry.


\begin{theorem}[Implicit in {\cite[Theorem~5.1]{FinkHKS17}}]
    Consider 
    determining whether any $d+1$ of the given $n$ points in $\bR^d$
    lie on a non-vertical $(d-1)$-dimensional hyperplane.
    Under the (real-valued) $k$-SUM conjecture, no randomized algorithm 
    solves this problem in time $n^{\ceil{(d+1)/2}-\gamma}$ for any $\gamma > 0$.
\end{theorem}

Their reduction is as follows. Let $P = (p^{(1)}, \dotsc, p^{(d)})$ be the
vertices of a convex polygon embedded in the hyperplane $\{x_d=0\}
\subset \bR^d$. Let $c = \frac{1}{d}\sum_{i \in [d]} p^{(i)}$ be the center
of mass of the vertices. For $j\in[d]$,
let $e_j$ denote the the
$j^{\text{th}}$ standard basis vector. 
Given a $(d+1)$-SUM instance $A_1, \dotsc,
A_{d+1}$, map each value $a_i \in A_i$ for $i \in [d]$ to the point $a_i^*
\define p^{(i)} + a_i \cdot e_d \in \bR^d$. Then map each value $a_{d+1} \in
A_{d+1}$ of the last input set to the point $a_{d+1}^* \define c +
\frac{a_{d+1}}{d} \cdot e_d \in \bR^d$. Let $S$ be the set of $(d+1)n$ points
given by this reduction. It follows that the $(d+1)$-SUM instance is a YES
instance iff there exist $d+1$ points in $S$ that lie on a
non-vertical $(d-1)$-dimensional hyperplane.

To prove our distribution approximation lower bound, we modify the above construction as follows:

\begin{enumerate}
    \item
    Multiply all points by $d$ to ensure integer coordinates.

\item Replace each point $a_i^*$ with two nearby points: one slightly above, labeled 1 by $f$; one slightly below, labeled 0.

\item Let $\cD$ be the uniform distribution over these points; we provide
    access to labeled 
    samples from $\cD$ and $f$ and query access to $f$.
\end{enumerate}

We will show that if we start with a YES instance of $(d{+}1)$-SUM, then some halfspace correctly
labels $(d{+}1)(n{+}1)$ points. If we start with a NO instance, then every halfspace labels at most
$\frac{2(d+1)n}{2} + d=  (d+1)(n+1) -1$ points correctly. Thus, approximating $\dist_\cD(f, \cH)$ to
additive error $O_d(1/n)$ solves the $(d{+}1)$-SUM problem. Setting $\eps = \Theta(\frac 1 n)$
yields \cref{thm:distribution-free-testing}.

We now formalize this argument.

\vspace{-5mm}
\begin{figure}[h] 
    \begin{minipage}[b]{0.4\textwidth}
    \centering
    \tdplotsetmaincoords{0}{90}
        \begin{tikzpicture}[tdplot_main_coords, scale=1.2]
            \foreach \x in {1,2,3} {
              \draw[tdplot_main_coords] (\x,0.1,0) -- (\x,-0.1,0);
            }
            
            \foreach \y in {1,2,3} {
              \draw[tdplot_main_coords] (0.1,\y,0) -- (-0.1,\y,0);
            }
            
            \draw[->, thick] (0,0,0) -- (4,0,0) node[anchor=north east]{};
            \draw[->, thick] (0,0,0) -- (0,4,0) node[anchor=north west]{};
            
            \coordinate (Z) at (0,0,0);
            \coordinate (X) at (3,0,0);
            \coordinate (Y) at (0,3,0);
            \coordinate (C) at (1,1,0);
            
            \filldraw[fill=blue!30, opacity=0.5, draw=black] (Z) -- (X) -- (Y) -- cycle;
            
            \filldraw (Z) circle (1.2pt) node[anchor=south] {$p^{(1)} = (0,0,0)$};
            \filldraw (X) circle (1.2pt) node[anchor=north west] {$p^{(2)} = (3,0,0)$};
            \filldraw (Y) circle (1.2pt) node[anchor=south] {$p^{(3)} = (0,3,0)$};
            \filldraw (C) circle (1.2pt) node[anchor=south] {$c = (1,1,0)$};
        \end{tikzpicture}
        \caption{The polygon in the reduction for $d=3$.}
        \label{fig:ksum-reduc-1}
    \end{minipage}%
    \hspace{0.1\textwidth}
    \begin{minipage}[b]{0.48\textwidth}
    \centering
    \tdplotsetmaincoords{60}{110}
        \begin{tikzpicture}[tdplot_main_coords, scale=1.2]
            \foreach \x in {1,2,3} {
              \draw[tdplot_main_coords] (\x,0.1,0) -- (\x,-0.1,0);
            }
            
            \foreach \y in {1,2,3} {
              \draw[tdplot_main_coords] (0.1,\y,0) -- (-0.1,\y,0);
            }
            
        
            \draw[->, thick] (0,0,0) -- (4,0,0) node[anchor=north east]{};
            \draw[->, thick] (0,0,0) -- (0,4,0) node[anchor=north west]{};
            
            \coordinate (Z) at (0,0,0);
            \coordinate (X) at (3,0,0);
            \coordinate (Y) at (0,3,0);
            \coordinate (C) at (1,1,0);
            
            \coordinate (Ztop) at (0,0,2.5);
            \coordinate (Zbot) at (0,0,-2.5);
            \coordinate (Xtop) at (3,0,2.5);
            \coordinate (Xbot) at (3,0,-2.5);
            \coordinate (Ytop) at (0,3,2.5);
            \coordinate (Ybot) at (0,3,-2.5);
            \coordinate (Ctop) at (1,1,2.5);
            \coordinate (Cbot) at (1,1,-2.5);
            
            \filldraw[fill=blue!30, opacity=0.5, draw=black] (Z) -- (X) -- (Y) -- cycle;
            
            \draw[red][dashed] (Zbot) -- (Ztop);
            \draw[red][dashed] (Xbot) -- (Xtop);
            \draw[red][dashed] (Ybot) -- (Ytop);
            \draw[red][dashed] (Cbot) -- (Ctop);
            
            \filldraw (Z) circle (1.2pt) node[anchor=east] {$p^{(1)}$};
            \filldraw (X) circle (1.2pt) node[anchor=south east] {$p^{(2)}$};
            \filldraw (Y) circle (1.2pt) node[anchor=south west] {$p^{(3)}$};
            \filldraw (C) circle (1.2pt) node[anchor=south east] {$c$};
            
            \filldraw[green!70!black] (0,0,2.2) circle (1.2pt) node[anchor=east] {$\beta(a_1^*)$};
            \filldraw[blue] (0,0,1.8) circle (1.2pt) node[anchor=east] {$\alpha(a_1^*)$};

            \filldraw[green!70!black] (0,0,1.2) circle (1.2pt) node[anchor=east] {$\beta(b_1^*)$};
            \filldraw[blue] (0,0,0.8) circle (1.2pt) node[anchor=east] {$\alpha(b_1^*)$};
            
            \filldraw[green!70!black] (3,0,-0.8) circle (1.2pt) node[anchor=east] {$\beta(a_2^*)$};
            \filldraw[blue] (3,0,-1.2) circle (1.2pt) node[anchor=east] {$\alpha(a_2^*)$};
            
            \filldraw[green!70!black] (3,0,2.2) circle (1.2pt) node[anchor=east] {$\beta(b_2^*)$};
            \filldraw[blue] (3,0,1.8) circle (1.2pt) node[anchor=east] {$\alpha(b_2^*)$};
            
            \filldraw[green!70!black] (0,3,2.2) circle (1.2pt) node[anchor=west] {$\beta(a_3^*)$};
            \filldraw[blue] (0,3,1.8) circle (1.2pt) node[anchor=west] {$\alpha(a_3^*)$};
            
            \filldraw[green!70!black] (0,3,-0.8) circle (1.2pt) node[anchor=west] {$\beta(b_3^*)$};
            \filldraw[blue] (0,3,-1.2) circle (1.2pt) node[anchor=west] {$\alpha(b_3^*)$};
            
            \filldraw[green!70!black] (1,1,-0.8) circle (1.2pt) node[anchor=west] {$\beta(a_4^*)$};
            \filldraw[blue] (1,1,-1.2) circle (1.2pt) node[anchor=west] {$\alpha(a_4^*)$};

            \filldraw[green!70!black] (1,1,-1.8) circle (1.2pt) node[anchor=west] {$\beta(b_4^*)$};
            \filldraw[blue] (1,1,-2.2) circle (1.2pt) node[anchor=west] {$\alpha(b_4^*)$};

            
            
            
            
            
            
            
        \end{tikzpicture}
        \caption{An illustration of our reduction for $d = 3$ and $n = 2$.}
    \label{fig:ksum-reduc-2}
  \end{minipage}
\end{figure}

\begin{proof}[Proof of \cref{thm:distribution-free-testing}]
    Given a $(d+1)$-SUM instance with input lists $A_1, \dots, A_{d+1}$, each of length~$n$, we
    reduce to distribution-free distance approximation
    for halfspaces over $\bZ^d$ with parameter $\epsilon = \Theta(\frac 1 n)$
    using the following construction.

    Let $P = (p^{(1)}, \dotsc, p^{(d)})$ be the polygon in the hyperplane $\{x_d = 0\} \subset
    \bR^d$, where $p^{(1)}$ is the zero vector and $p^{(i)} \define d \cdot e_{i-1}$  for all $i =
    2, \dotsc, d$. Let $c$ be the center of mass of the vertices of $P$, \ie $c = \frac{1}{d}
    \sum_{i=1}^d p^{(i)} = (1,\dots,1,0)$. Initialize a point set $S = \emptyset$ and a function $f
    : \bZ^d \to \zo$ to be zero on the entire domain. For each $i \in [d]$ and $a_i \in A_i$, define
    $a_i^* \define p^{(i)} + 4d a_i \cdot e_d$; for the last set $A_{d+1}$, for each $a_{d+1} \in
    A_{d+1}$, define $a_{d+1}^* \define c + 4a_{d+1} \cdot e_d$. For each $i \in [d+1]$ and $a_i \in
    A_i$, add two points to the set $S$: $$\alpha(a_i^*) \define a_i^* - e_d, \text{ \ and \  }
    \beta(a_i^*) \define a_i^* + e_d,$$ and set the labels of these points to $f(\alpha(a_i^*))
    \define 0$ and $f(\beta(a_i^*)) \define 1$. This completes the construction of the set $S$ of
    $2(d+1)n$ points and of the Boolean function $f : \bZ^d \to \zo$.

    To finish the reduction, let $\cD$ be the uniform distribution over $S$  and set
    $\eps= \frac 1{5(d+1)n}$. If the distance approximation algorithm executed on the constructed
    instance returns an estimate at most $\frac{|S|-(n+1)(d+1) +1/2}{|S|}$, accept; otherwise,
    reject.
    
    It remains to prove the correctness of the reduction and analyze the efficiency of the
    simulation.

    \textbf{Correctness.}
    For each list $A_i$ in the $k$-SUM instance, all points created from the numbers in $A_i$
    using $\alpha$ and $\beta$ functions are on the same line (see \Cref{fig:ksum-reduc-2});
    specifically, they only differ in the coordinate $d$. For all $i\in[d+1]$ and all $a_i,b_i\in
    A_i$ such that $a_i<b_i$, we have $a_i\leq b_i-1$, since $a_i$ and $b_i$ are integers.
    Consequently, the points $a_i^*$ and $b_i^*$ differ by at least $4$ in coordinate $d$, which
    implies $\alpha(a_i^*)\prec\beta(a_i^*)\prec\alpha(b_i^*)\prec\beta(b_i^*)$.

    Now suppose that $A_1, \dotsc, A_{d+1}$ is a YES instance of the
    $(d+1)$-SUM problem. Then there exist $a_1 \in A_1,\dots,a_{d+1} \in A_{d+1}$  such that
    $\sum_{i=1}^d a_i = a_{d+1}$, which implies
    $\frac{1}{d} \sum_{i=1}^d a_i^* = a_{d+1}^*$. Hence, the unique
    hyperplane $H$ passing through points $a_1^*, \dotsc, a_d^*$ also passes
    through point $a_{d+1}^*$. Therefore, there exists a halfspace $h$
    which agrees with $f$ on at least all of the points $\alpha(a_i^*),
    \beta(a_i^*)$ for $i \in [d+1]$, plus at least half of the remaining
    $\alpha(\cdot), \beta(\cdot)$ points in $S$. Hence, $f$ agrees with $h$
    on at least
    $\frac{2(d+1)n}{2} + (d+1) = (d+1)(n+1)$ points of $S$, and thus
    $\dist_\cD(f, \cH) \le \frac{|S|-(n+1)(d+1)}{|S|}$.

    Next, suppose that $A_1, \dotsc, A_{d+1}$ is a NO instance
    of the $(d+1)$-SUM problem. We claim that
    \[
        \dist_\cD(f, \cH) \ge \frac{|S|-(n+1)(d+1) +1}{|S|} \,.
    \]
    The claim follows if we
    show that, for every halfspace $h$, the function
    $f$ agrees with $h$ on at most $\frac{2(d+1)n}{2} + d = (d+1)
    (n+1)-1$ points of $S$. Suppose for contradiction that $f$ and $h$ agree on at
    least $(d+1)(n+1)$ points of $S$. 
    Then for each $i \in [d+1]$, there must be some $a_i \in A_i$ such that $h$ correctly labels
    both $\alpha(a_i^*)$ and $\beta(a_i^*)$. This implies that the separating hyperplane $H$ passes
    between these two points for each $i\in[d+1]$.

    Since $A_1, \dotsc, A_{d+1}$ are sets of integers and form a NO instance of the
    $(d+1)$-SUM problem, 
    \[
        \abs*{\sum_{i=1}^d a_i - a_{d+1}} \ge 1 \, , 
  \text{\ \  and thus \ \  }
        \abs*{\frac{1}{d} \sum_{i=1}^d 4d a_i - 4a_{d+1}} \ge 4 \,.
    \]
    Therefore, for every choice of $\lambda_i \in [-1,1]$ for each $i \in [d+1]$,
    the triangle inequality implies that
    \[
        \abs*{\frac{1}{d} \sum_{i=1}^d (4d a_i + \lambda_i) - (4a_{d+1} + \lambda_{d+1})}
        \ge 2 > 0 \,.
    \]
    Consequently, no hyperplane can pass within vertical distance 
    $1$ of every point $a_i^*$ for all $i \in [d+1]$, which contradicts the
    previous conclusion that $H$ passes between  points $\alpha(a_i^*)$ and
    $\beta(a_i^*)$ for each $i \in [d+1]$, completing the proof of the claim.

    We conclude that, if the distance approximation algorithm is run with parameter $\eps= \frac
    1{5(d+1)n}<\frac 1 {2|S|}$, the reduction algorithm correctly answers for both YES and NO
    instance with probability at least~$\frac 23$.

    \textbf{Efficiency.} To simulate a distance approximation algorithm $\cA$,
    we first read the entire input $A_1, \dotsc A_{d+1}$ and sort each list
    $A_i$ in increasing order. This takes $O(n \polylog(n))$ time.
   To answer a query $f(x)$ from $\cA$, we spend $O(\polylog(n))$ time to
    perform binary search and check whether $x = \alpha(a_i^*)$ or $x =
    \beta(a_i^*)$ for some $a_i \in A_i$; if so, we return the corresponding label, and otherwise return $0$.  When $\cA$ requests a labeled sample $(x \sim
    \cD, f(x))$, we use random bits to take a sample from the point set $S$ and
    compute its label as discussed. This also takes $O(\polylog(n))$ time.

    Suppose $\cA$ runs in time $O\left((1/\epsilon)^{\ceil{(d+1)/2} -
    \gamma}\right)$. Then we can solve our $(d+1)$-SUM instance in time
    $O\left((1/\epsilon)^{\ceil{(d+1)/2} - \gamma}
    \polylog(1/\epsilon)\right)$, a contradiction. This completes the proof.
\end{proof}

\begin{remark}[On the use of integrality]
    \label{remark:real-values}
    Aside from our choice of RAM model, the proof of \cref{thm:distribution-free-testing} also
    relies on the integrality of the input to ensure that a NO instance of $(d{+}1)$-SUM maps to a
    point set that cannot be well-approximated by any halfspace. Specifically, any hyperplane that
    does not pass through the exact point $a_i^*$ cannot lie sufficiently close it to still separate
    $\alpha(a_i^*)$ and $\beta(a_i^*)$, preventing “almost-satisfiable” instances.

    If we replicated our construction in the real-valued setting, with an
    appropriately adapted real RAM model, the only caveat is that, under
    arbitrary real-valued inputs, we no longer obtain this robust separation
    property for free; although we could place the points
    $\alpha(a_i^*)$ and $\beta(a_i^*)$ arbitrarily close to $a_i^*$, it seems
    difficult to choose this distance a priori without solving a $k$-SUM-like
    problem.

    However, this issue disappears if we assume that the real-valued
    $k$-SUM problem, in the real RAM model, is still difficult even if the input
    is promised to be integer. Indeed, not only does this seem natural
    enough (the real RAM model is not ``supposed'' to distinguish between
    integers and non-integers), but the original work of \cite{GajentaanO95},
    who systematized the notion of 3SUM-hardness, made precisely the
    assumption that 3SUM is difficult on integer inputs in the real RAM model.
    We therefore view the real-valued analogue of \cref{thm:distribution-free-testing} as equally
    plausible.
\end{remark}

\section{Statistical query lower bounds for the Gaussian distribution}
\label{sec:sq-lb}

In this section, we present \emph{evidence} for a time complexity lower bound
for distance approximation of halfspaces under the standard Gaussian
distribution over $\bR^d$, in the low-dimensional regime where $d$ is a
constant, in the form of a sample complexity lower bound against Statistical
Query (SQ) algorithms.


For a fixed $d \in \bN$, let $\cH$ denote the class of halfspaces over $\bR^d$,
each of which is an $\bR^d \to \pmset$ function, and $\cN(0,I)$ denote the standard
multivariate Gaussian distribution over $\bR^d$.

\begin{definition}[Distribution-specific distance approximation in the SQ model]
    Let $d \in \bN$. A randomized SQ algorithm $\cA$ is a \emph{distance
    approximation} algorithm for halfspaces under $\cN(0, I)$ if, given
    parameters $\epsilon, \delta \in (0,1)$ and access to input function $f :
    \cX \to \pmset$ via a $\Stat$ oracle (see \cref{def:stat-sq}), algorithm
    $\cA$ outputs a number $\widehat{\alpha}$ which, with probability at least
    $1-\delta$, satisfies $\abs{\dist_{\cN(0,I)}(f, \cH) - \widehat{\alpha}} \le
    \epsilon$.
\end{definition}

The main result of this section is the following lower bound, which
qualitatively matches the fine-grained, distribution-free lower bound in
\cref{thm:distribution-free-testing}.

\begin{restatable}{theorem}{thmsqlb}
    \label{thm:sq-lb}
    Let $d \geq 2$ be a constant in $\bN$. Then every randomized SQ distance
    approximation algorithm for halfspaces under $\cN(0, I)$ with success
    probability at least 0.51 requires\footnote{By $\Omega(d)$, we mean $cd$,
    where $c$ is an absolute constant.} $(1/\eps)^{\Omega(d)}$ queries to
    $\Stat(\eps^{\Omega(d)})$.
\end{restatable}

We prove \cref{thm:sq-lb} by extending an argument of \cite{DKZ20}, who showed
an SQ lower bound of $d^{\poly(1/\epsilon)}$ for agnostic learning of halfspaces
under the Gaussian distribution in the high-dimensional regime where $d \gg
\poly(1/\epsilon)$. Our argument combines a packing number result for the unit
sphere in the low-dimensional regime, the construction of a ``pseudorandom''
function that serves as the \textsc{No} instance for the distance approximation
task, and a nonuniform derandomization argument.

The rest of this section is organized as follows. \cref{section:sq-model}
defines the SQ model and introduces the key notion of SQ dimension.
\cref{sec:prior-work-sq} extracts the result we use from \cite{DKZ20}.
\cref{sec:packing} establishes the packing result. \cref{sec:pseudorandom}
constructs the pseudorandom function for SQ algorithms. \cref{sq:sq-det}
combines these ingredients to give a set of functions with high SQ dimension,
which implies a lower bound against deterministic SQ algorithms. Finally,
\cref{sq:sq-rand} proves \cref{thm:sq-lb}.

\subsection{SQ algorithms and SQ dimension}
\label{section:sq-model}

We start by formally defining the SQ model.

\begin{definition}[$\Stat$ oracle, SQ algorithm \cite{Kearns98}]
    \label{def:stat-sq}
    Let $\cX$ be a domain, $f : \cX
    \to [-1,1]$ be a function, $\cD$ be a probability distribution over $\cX$,  and $\tau > 0$ be a \emph{tolerance parameter}.
    Given a \emph{statistical query} 
    $q : \cX \times [-1,1] \to
    [-1,1]$, the oracle $\Stat(\tau)$ outputs $v \in \bR$ satisfying $\Big|\Exu{x \sim
    \cD}{q(x, f(x))} - v\Big| \le \tau$.
    A \emph{Statistical Query (SQ) algorithm}  accesses its input $f$ only via
calls to a $\Stat$ oracle.
\end{definition}

Given a domain $\cX$, distribution $\cD$ over $\cX$, and
functions\footnote{For simplicity, we omit from our statements tedious remarks
about measurability of functions.} $f, g : \cX \to \bR$, we call $\Exu{x \sim
\cD}{f(x) g(x)}$ the \emph{correlation} between $f$ and $g$. In our setting,
$\cD$ is  the standard multivariate Gaussian.

\begin{definition}[SQ dimension \cite{BlumFJKMR94}]
    Let $\cX$ be a domain, $\cF$ be a
    class of $\cX \to [-1,1]$ functions, and $\cD$ be a probability distribution over $\cX$. The \emph{SQ
    dimension of $\cF$ under
    $\cD$}, denoted $\SQDim(\cF, \cD)$, is the largest $s\in\bN$ such that there exist distinct functions
    $f_1, \dotsc, f_s \in \cF$ satisfying $\Big|\Exu{x \in \cD}{f_i(x) f_j(x)}\Big| \le \frac 1 s$ for all 
    distinct $i,j\in[s]$.
\end{definition}

Blum et al.~\cite{BlumFJKMR94} showed that a lower bound on $\SQDim(\cF, \cD)$
implies a lower bound for SQ algorithms that  weakly learn $\cF$.
Concretely, we have the following result, whose simplified proof by
\cite{Szorenyi09} we  adapt to obtain a lower bound for distance
approximation.

\begin{restatable}[SQ dimension bound for weak learning]
{theorem}{thmsqweaklearning}
    \label{thm:sq-weak-learning}
    Let $\cX$ be a domain,  $\cF$ be a class of $\cX \to [-1,1]$ functions, and
    $\cD$ be a probability distribution over $\cX$ with $\SQDim(\cF, \cD) = s$.
    Then every (deterministic) SQ algorithm that, on input function $f \in \cF$,
    outputs a function $g : \cX \to [-1,1]$ whose correlation with $f$ is
    $\Exu{x \sim \cD}{f(x) g(x)} \ge s^{-1/3}$ requires at least
    $\frac{s^{1/3}}2 - 1$ queries to $\Stat(s^{-1/3})$.
\end{restatable}

\paragraph*{Correlation queries.}
As observed by \cite{BshoutyF02}, in the distribution-specific setting, it
suffices to consider statistical queries of the form $q(x, f(x)) = g(x) f(x)$
with $g : \cX \to [-1,1]$. Given $g$, the oracle returns a value $v$ satisfying
$\big|\Exu{x \sim \cD}{g(x) f(x)} - v\big| \le \tau$. Any general query can be
simulated with two such queries, so we restrict attention to this form.

\subsection{Prior work: high SQ dimension from packing numbers on the sphere}
\label{sec:prior-work-sq}

Diakonikolas, Kane and Zarifis \cite{DKZ20} 
proved an SQ lower bound of $d^{\poly(1/\epsilon)}$ for weakly learning halfspaces over $\bR^d$ under the Gaussian distribution in high dimensions. Their proof constructs a set 
of Boolean functions with large SQ dimension, each correlated with a halfspace. This is done by selecting a large set of nearly orthogonal unit vectors in $\bR^d$ and mapping each vector $u$ to a Boolean function
by projecting along $u$ a one-dimensional $k$-alternating function
satisfying a certain moment-matching condition.

Our proof builds upon the construction of \cite{DKZ20}, but we use a different
packing bound suited to the low-dimensional setting. We start by extracting the
following result from \cite{DKZ20}.

\begin{theorem}[Implicit in \cite{DKZ20}]
    \label{thm:dkz20-implicit}
    Let $d, n, k \in \bN$ and $\rho \in (0,1)$. Suppose there exist vectors
    $u_1, \dotsc, u_n \in \bS^{d-1}$ satisfying $\abs*{\inp{u_i}{u_j}} \le \rho$ for 
    all distinct $i,j\in[n]$.
    Then there exist functions $f_1, \dotsc, f_n : \bR^d
    \to \pmset$ satisfying the following conditions.
    \begin{enumerate}
        \item For each $i \in [n]$, there exists a halfspace $h_i : \bR^d \to
        \pmset$ such that $\Exu{x \sim \cN(0,I)}{f_i(x) h_i(x)} \ge \frac{1}{2k}$.
        \item 
        $\Big|\Exu{x \sim
        \cN(0,I)}{f_i(x) f_j(x)}\Big| \le 2\rho^{k+1}$ for all distinct $i,j\in[n]$.
    \end{enumerate}
\end{theorem}

\subsection{Packing slightly uncorrelated vectors on the low-dimensional sphere}
\label{sec:packing}

We begin with a result from~\cite{CFJ13} on the asymptotic distribution of the
minimum and maximum angles among $n$ vectors drawn independently and uniformly
from the unit sphere $\bS^{d-1}$, for constant dimension $d$. In contrast,
\cite{DKZ20} (via a lemma from \cite{DKS17}) relies on a different result
from~\cite{CFJ13} applicable to the high-dimensional setting.

\begin{theorem}[Theorem 2 of \cite{CFJ13}]
    \label{thm:theta-min-max}
    Let $d \ge 2$ be an integer. Let $u_1 \dotsc, u_n$ be sampled independently and
    uniformly at random from $\bS^{d-1}$, and let $\theta_{\min}, \theta_{\max}$ denote
    the minimum and maximum angles, respectively, between any pairs of vectors $u_i, u_j$
    for distinct $i,j\in[n]$ 
    Then as $n \to \infty$, the random variables $n^{2/(d-1)}
    \theta_{\min}$ and $n^{2/(d-1)} (\pi - \theta_{\max})$ converge weakly to the
    distribution with cumulative distribution function (CDF)
    \[
        F(x) = \begin{cases}
            1 - e^{-K_d x^{d-1}} & \text{if } x \ge 0 \,, \\
            0                    & \text{if } x < 0 \,,
        \end{cases}
    \]
    where $K_d > 0$ is a constant that depends only on $d$.
\end{theorem}

\Cref{thm:theta-min-max} establishes weak convergence of distributions, \ie
pointwise convergence of the CDFs $F_n$ to $F$. The following standard fact
shows that, since $F$ is continuous, convergence is  uniform: for each $\delta >
0$, we can choose large enough $n$ so that $\abs*{F_n(x) - F(x)} \le \delta$ for
all $x$.

\begin{fact}
    \label{fact:uniform-convergence-cdf}
    Suppose $(X_n)_{n \in \bN}$ is a sequence of real-valued random variables
    converging weakly to $X$. Let $F_n$ and $F$ be the CDFs of $X_n$ and  $X$,
    respectively. If $F$ is continuous, then $F_n \to F$ uniformly.
\end{fact}

We now prove our packing lemma.

\begin{lemma}[Packing slightly uncorrelated vectors on the sphere]
    \label{lemma:packing}
    Let $d \ge 2$ be a fixed integer. Then for all constant $a > 0$ and $b \in
    \left(0, \frac{d-1}{4d}\right)$, and all sufficiently small $\epsilon > 0$,
    there exists a set $S$ of at least $(1/\epsilon)^{b d}$ vectors on
    $\bS^{d-1}$ such that, for all distinct $u, v \in S$, it holds that
    $\abs*{\inp{u}{v}} \le \epsilon^{a \epsilon}$.
\end{lemma}
\begin{proof}
    Let $n \define
    \ceil{(1/\epsilon)^{b d}} \le 2 (1/\epsilon)^{b d}$ and $S = (u^1, \dotsc, u^n)$ be a sequence of  points sampled uniformly and independently from $\bS^{d-1}$. We show that with
    positive probability, every pair $u^i \ne u^j$ in $S$ with $i \ne j$
    satisfies $\abs*{\inp{u^i}{u^j}} \le \epsilon^{a\epsilon}$, thus
    establishing the existence of such a set $S$.

    Let $u, v \in \bS^{d-1}$ be distinct and  $\theta \in [0, \pi]$ be the angle between them.
    Let
    $\alpha \define \epsilon^{a \epsilon}$.  
    Since $u, v$ are
    unit vectors, $\abs*{\inp{u}{v}} = \abs*{\cos \theta}$. Thus,
    \begin{equation}
        \label{eq:arccos-cond}
        \abs*{\inp{u}{v}} \le \alpha
        \iff \abs*{\cos \theta} \le \alpha
        \iff \theta \in [\arccos(\alpha), \arccos(-\alpha)] \,.
    \end{equation}
    We lower bound $\alpha$ by
    $
        \alpha
        = \epsilon^{a \epsilon}
        = e^{-a \epsilon \ln(1/\epsilon)}
        \ge 1 - a \epsilon \ln\left(\frac 1\epsilon\right) \,.
    $
    Since $\arccos$ is a decreasing function, a sufficient condition for \eqref{eq:arccos-cond}
    to hold is 
    \[
        \theta \in
            [\arccos(1 - a \epsilon \ln(1/\epsilon)),
            \arccos(-1 + a \epsilon \ln(1/\epsilon))] \,.
    \]
    By the series expansions for $x \to 0^+$,
    \begin{align*}
        \arccos(1 - x) &= \sqrt{2x} + O(x^{3/2}) \le 2\sqrt{x} \,, \\
        \arccos(-1 + x) &= \pi - \sqrt{2x} - O(x^{3/2}) \ge \pi - 2\sqrt{x}.
    \end{align*}
    Thus, for sufficiently small $\epsilon$, it suffices  if $\theta \in [C, \pi - C]$ where $C \define 2\sqrt{a \epsilon \ln(1/\epsilon)}$. Equivalently,
    \begin{align}
        \label{eq:theta-cond}
        \theta\geq C \text{ and } \pi - \theta \ge C \,.
    \end{align}
    By \cref{thm:theta-min-max}, all $u^i, u^j \in S$ with $i \ne j$ satisfy
    \eqref{eq:theta-cond} and hence \eqref{eq:arccos-cond}, except with probability at most
    $\Pr{\theta_{\min} < C}
                + \Pr{\pi - \theta_{\max} < C}$ 
        \begin{align}\label{eq:theta-probs}
           =
                \Pr{n^{2/(d-1)} \theta_{\min}
                    <  n^{2/(d-1)} C}
                + \Pr{n^{2/(d-1)} (\pi - \theta_{\max})
                    <  n^{2/(d-1)} C}. 
        \end{align}
    For sufficiently small $\epsilon > 0$, the RHS in the inequalities above is
    \begin{align*}
        n^{2/(d-1)}C
        &\le 2 \left(2 (1/\epsilon)^{b d}\right)^{2/(d-1)} \sqrt{a \epsilon \ln(1/\epsilon)} \\
        &=
            2 \sqrt{a} \cdot
            4^{1/(d-1)} \cdot
            \epsilon^{\frac{1}{2} - \frac{2 b d}{d-1}}
            \sqrt{\ln(1/\epsilon)}
        \le O(1) \cdot \epsilon^{\Omega(1)} \,,
    \end{align*}
    where the first step uses the bound $n\le 2 (1/\epsilon)^{b d}$ and the value of $C$, and the last step uses the assumption that $b < (d-1)/4d$.

    By \cref{thm:theta-min-max,fact:uniform-convergence-cdf}, for all sufficiently small $\epsilon$
    (\ie all sufficiently large $n$), each of the probabilities from \eqref{eq:theta-probs} is
    approximated by the limit CDF
    $F$ from \cref{thm:theta-min-max} up
    to an additive error of (say) $0.1$. Thus, each of the probabilities from \eqref{eq:theta-probs}
    is at most
    \[
        0.1 + F(O(1) \cdot \epsilon^{\Omega(1)})
        = 0.1 + 1 - e^{-K_d (O(1) \cdot \epsilon^{\Omega(1)})^{d-1}}
        = 0.1 + 1 - e^{-O(1) \cdot \epsilon^{\Omega(1)}}
        < 0.11 \,.
    \]
    Thus, $S$ satisfies the desired property with positive probability.
\end{proof}

\subsection{A pseudorandom function for SQ algorithms}
\label{sec:pseudorandom}

Combining \cref{thm:dkz20-implicit,lemma:packing}, we obtain a large set $\cF$
of Boolean functions with low pairwise correlations, where each $f \in \cF$ is
correlated with a halfspace. To derive an SQ lower bound for distance
approximation (as in the proof of \cref{thm:sq-weak-learning} for weak learning
by \cite{Szorenyi09}), we would like to argue that an SQ algorithm cannot
distinguish such an input $f$ from ``random noise''\!\!, because an adversarial
\textsc{Stat} oracle can consistently answer 0 until many queries are made.

However, we wish to obtain a lower bound for distance approximation of
\emph{functions}, not of randomized labelings (\ie joint point-label
distributions). Intuitively, a
random function $f : \bR^d \to \pmset$ is indistinguishable from noise, but
formally, such a function may not even be measurable, so it would be unsuitable
as an input in the SQ model. To handle this, we construct a
``pseudorandom'' function $f_0$ that is sufficiently uncorrelated with every
halfspace and every query of a target deterministic algorithm $\cA$. To
construct such $f_0$, we argue that the set of queries algorithm $\cA$ makes on
all inputs is finite and use this fact together with the connection between
between VC dimension and covering numbers (for the class of halfspaces) to
discretize the space $\bR^d$ into sufficiently small cells, and then make $f_0$
balanced within each cell.

Our argument requires the queries of $\cA$ to come from a finite set, regardless
of the oracle's answers. To enforce this, we define a discretized oracle as
follows. For all $x \in \bR$ and $\tau \in (0,1)$, let $\round_\tau(x)$ denote
the integer multiple of $\tau$ that is closest to $x$, with rounding towards
zero; that is, $\round_\tau(x) \define \sign(x) \cdot \tau \cdot
\floor{\abs{x}/\tau}$. A \emph{$\tau$-rounding oracle} answers query $g : \cX
\to [-1,1]$ with the quantity $\round_\tau\Big(\Exu{x \sim \cD}{g(x)
f(x)}\Big)$. Such an oracle is a valid $\Stat(\tau)$ oracle. If $\Exu{x \sim
\cD}{g(x) f(x)} \in (-\tau, \tau)$, then the $\tau$-rounding oracle outputs $0$,
implementing the desired adversarial behavior.

\begin{lemma}
    \label{lemma:f0}
    Let $d \in \bN$ and $\tau, \epsilon > 0$. Let $\cA$ be a deterministic SQ
    algorithm that takes a bounded number of bits of advice\footnote{We allow
    advice because in the final stage of the proof of \cref{thm:sq-lb}, we
    convert a randomized algorithm to a deterministic algorithm with advice.}
    and makes a bounded number of queries to a $\tau$-rounding oracle. Then
    there exists a function $f_0 : \bR^d \to \pmset$ such that:
    \begin{enumerate}
        \item  $\dist_{\cN(0,I)}(f_0, \cH) \ge \frac 12 - \frac \epsilon{100}$.
        \item For each function $f : \bR^d \to \pmset$, every query $g : \bR^d
        \to [-1,1]$ made by $\cA$ on input $f$ satisfies $\Big|\Exu{x \sim
            \cN(0,I)}{g(x) f_0(x)}\Big| \le \frac \tau 2$.
    \end{enumerate}
\end{lemma}
\begin{proof}
    Since each output of the $\tau$-rounding oracle comes from the finite set of
    integer multiples of $\tau$ in $[-1,1]$, and $\cA$ takes a finite number of
    bits of advice and terminates after finitely many queries, there are
    finitely many sequences of oracle outputs seen by $\cA$ regardless of its
    input. Since $\cA$ is deterministic, there exists a finite set
    $\cG$ such that every query $g$ made by $\cA$ comes from $\cG$.

    Now, since the class $\cH$ of halfspaces over $\bR^d$ has finite VC
    dimension, a standard result in learning theory
    \cite[Corollary~1]{Haussler95} says that $\cH$ has finite covering number
    with respect to the Gaussian distribution, that is, there exists a finite
    set $\cR$ of reference halfspaces such that every halfspace $h$ in
    $\bR^d$ satisfies $\dist_{\cN(0,I)}(h, r) \le \epsilon/100$ for some
    $r \in \cR$.
    
    For each $x \in \bR^d$, let its \emph{color} $C(x)$ be the tuple
    containing all labels of $x$ by the (rounded) query functions from $\cG$ and the reference halfspaces from $\cR$:
    \[
        C(x) \define \left(
            \left(\round_{\tau/2}(g(x))\right)_{g \in \cG},
            \left(r(x)\right)_{r \in \cR}
        \right) \,.
    \]
    The number of possible colors is finite. For each such color $c$, definite
    set $P_c = \{x \in \bR^d : C(x) = c\}$. Then sets $P_c$ partition $\bR^d$.
    For each part $P_c$, let $L_c \cup R_c$ be a partition of $P_c$ of balanced
    Gaussian measure, \ie a partition satisfying $\Pru{x \sim \cN(0,I)}{x \in
    L_c} = \Pru{x \sim \cN(0,I)}{x \in R_c}$. Define the function $f_0$ as
    follows: for all colors $c$ and $x \in P_c$, let $f_0(x) = +1$ if $x \in
    L_c$ and $f_0(x) = -1$ if $x \in R_c$.

    Now we prove the two items of \Cref{lemma:f0}. For the first item, let $h$
    be a halfspace in $\bR^d$. By construction, $f_0$ is $\frac 12$-far from
    every reference halfspace $r \in \cR$ under the Gaussian distribution, since
    for each part $P_c$, halfspace $r$ gives all of $P_c$ the same label,
    whereas $f_0$ gives label $+1$ or $-1$ with equal conditional probability.
    Using the covering property, let $r \in \cR$ be a reference halfspace
    satisfying $\dist_{\cN(0,I)}(h, r) \le \epsilon/100$. Then, by the triangle
    inequality,
    \[
        \dist_{\cN(0,I)}(f_0, h)
        \ge \dist_{\cN(0,I)}(f_0, r) - \dist_{\cN(0,I)}(h, r)
        \ge \frac{1}{2} - \frac \epsilon{100} \,.
    \]
    Hence $\dist_{\cN(0,I)}(f_0, \cH) \ge \frac{1}{2} - \frac \epsilon{100}$, as claimed.

    For the second item, let function $g \in \cG$ be a query of $\cA$. By the
    construction of $f_0$, 
    \begin{equation}\label{eq:correlation-of-f0-with-rounded-functions}        
        \Exu{x \sim \cN(0,I)}{\round_{\tau/2}(g(x)) f_0(x)} = 0 \,,
    \end{equation}
    because for each part $P_c$, every $x \in P_c$ has the same value of
    $\round_{\tau/2}(g(x))$, whereas $f_0$ gives label $+1$ or $-1$ with equal
    conditional probability. By applying
    \eqref{eq:correlation-of-f0-with-rounded-functions} and then Jensen's
    inequality together with the fact that $|f_0(x)|=1$ for all $x\in\bR^d$, we
    get
    \begin{align*}
        \abs*{\Exu{x \sim \cN(0,I)}{g(x) f_0(x)}}
        &= \abs*{\Exu{x \sim \cN(0,I)}{
            \left(g(x) - \round_{\tau/2}(g(x))\right) f_0(x)}} \\
        &\le \Exu{x \sim \cN(0,I)}{\abs*{g(x) - \round_{\tau/2}(g(x))}}
        \le \frac \tau 2 \,. \qedhere
    \end{align*}
\end{proof}

\subsection{Lower bound for deterministic algorithms}
\label{sq:sq-det}

We combine the ingredients above to prove a lower bound for deterministic SQ algorithms for distance
approximation of halfspaces. Our proof builds on the simplified argument of \cite{Szorenyi09} for
weak learning (\cref{thm:sq-weak-learning}), which uses the correlation bound in the definition of
SQ dimension to show that each statistical query, if answered adversarially, can rule out only a
limited number of inputs.

\begin{theorem}
    \label{thm:sq-lb-det}
    Let $d \geq 2$ be a constant. Then every deterministic SQ distance approximation algorithm for
    halfspaces under $\cN(0, I)$ with a bounded number of bits of advice requires $(1/\eps)^{\Omega(d)}$ queries to
    $\Stat(\eps^{\Omega(d)})$.
\end{theorem}
\begin{proof}
    Let $\gamma \in (0, 1)$ be a constant. First, we apply \cref{lemma:packing} with parameters $a =
    d$ and $b = \frac{\gamma (d-1)}{4d}$. This yields a set of $s = \left( 1/\eps
    \right)^{\frac{\gamma}{4} (d-1)}$ unit vectors with pairwise correlation at most $\rho
    = \eps^{d\eps}$. We apply \cref{thm:dkz20-implicit} with parameter $k = \frac{1}{\eps} - 1$ to
    this set of vectors to get a set $\cF$ of $s$  functions of the form $\bR^d \to \{ \pm 1 \}$
    satisfying the following conditions.
    \begin{enumerate}
        \item For each $f\in\cF$,
            there is a halfspace $h : \bR^d \to \{ \pm 1 \}$ satisfying
            $\displaystyle\Exu{x \sim \cN(0,I)}{f(x) h(x)} \ge \frac{1}{2k} \geq \frac{\eps}{2},$
            and consequently, $\dist_{\cN(0, I)} (f, \cH)$ is at most
            \[
                \dist_{\cN(0, I)} (f, h)
                = \Pru{x \sim \cN(0,I)}{f(x) \neq h(x)}
                = \frac 12\Big(1 - \Exu{x \sim \cN(0,I)}{f(x) h(x)}\Big)
                \leq \frac{1}{2} - \frac{\eps}{4} \,.
            \]

        \item\label{condition:pairs-of-functions} For all
            distinct $f,f'\in \cF$, we have
            $\Big|\Exu{x \sim \cN(0,I)}{f(x) f'(x)}\Big| \le 2\rho^{k+1} = 2\epsilon^d \leq
            \eps^{\frac{\gamma}{4} (d-1)} = \frac{1}{s}$, where the second inequality holds for
            small enough $\eps > 0$.
    \end{enumerate}
    Let $\cA$ be a deterministic SQ algorithm that makes less than $\frac{s^{1/3} - 1}{2}$ queries
    to a $\tau$-rounding oracle with $\tau = s^{-1/3}$. We apply \cref{lemma:f0} to get a function
    $f_0$ that has correlation at most $\frac{\tau}{2}$ with all  queries made by $\cA$ and
    satisfies  $\dist_{\cN(0, I)} (f_0, \cH) \geq \frac{1}{2} - \frac{\eps}{100}$. We define our
    ``adversarial'' oracle to always respond with 0 to queries made by $\cA$. Then $f_0$ is
    consistent with all query answers. We will show that some $f_*\in\cF$ is also consistent with
    all query answers. Then $\cA$ cannot distinguish $f_0$ from $f_*$, thereby failing to
    approximate the distance to $\cH$ well.

    We use the inner product notation $\inp{f}{g} \define \Exu{x \sim \cN(0, I)}{f(x) g(x)}$. Fix a
    query $g : \bR^d \to [0,1]$. By definition, it has norm $\| g \| \leq 1$.
    Let the \emph{bad set} $B := \{ f\in\cF : \inp g f \geq \tau \}$.
    Then $\inp{g}{\sum_{f \in B} f} \ge |B|\tau$.
    On the other hand, by the Cauchy-Schwarz inequality and
    condition~\ref{condition:pairs-of-functions} on the set $\cF$,
    \[
        \inp{g}{\sum_{f \in B} f}^2
        \le \left\| \sum_{f \in B} f \right\|^2
        = \sum_{f,f' \in B} \inp{f}{f'}
        \le \sum_{f \in B} \left( 1 + \frac{|B|-1}{s} \right)
        \le |B| + \frac{|B|^2}{s} \,.
    \]
    Combining the two inequalities and then substituting $\tau = s^{-1/3}$, we obtain $|B| \le
    \frac{s}{s\tau^2 - 1}=\frac s{s^{1/3}-1}$. Similarly, at most $\frac{s}{s^{1/3} - 1}$ functions
    in $\cF$ have correlation at most $-\tau$ with $g$. Hence, at most $\frac{2s}{s^{1/3} - 1}$ functions in $\cF$ are inconsistent with the answer 0 to query $g$. Since $\cA$ makes less
    than $\frac{s^{1/3} - 1}{2}$ queries, some $f_* \in \cF$ is consistent with every query answer
    being $0$. By construction, $f_0$ is also consistent with every query answer being $0$. However,
    $\dist_{\cN(0, I)} (f_0, \cH) \geq \frac{1}{2} - \frac{\eps}{100}$, whereas $\dist_{\cN(0, I)}
    (f_*, \cH) \leq \frac{1}{2} - \frac{\eps}{4}$. Thus $\cA$ cannot approximate distance to $\cH$
    with parameter $\frac\eps{16}$. Because $s^{1/3} = \left( 1/\eps \right)^{\frac{\gamma}{12}
    (d-1)}$, 
    this concludes our proof.
\end{proof}

\subsection{Lower bound for randomized algorithms}
\label{sq:sq-rand}

We now extend the lower bound in \Cref{thm:sq-lb-det} to randomized SQ algorithms via a nonuniform
derandomization argument. The key observation is that the lower bound for deterministic algorithms from \cref{thm:sq-lb-det} applies to \emph{nonuniform} algorithms. Since a randomized algorithm with
sufficiently high success probability implies, via a probabilistic argument, a deterministic algorithm with advice, we obtain a lower bound for randomized algorithms. A similar observation for weak learning algorithms appears in \cite{BshoutyF02}.

\begin{proof}[Proof of \cref{thm:sq-lb}]
    Any algorithm with success probability at least $0.51$ can be boosted to succeed with
    probability $1- \delta$ by running the algorithm $O(\log (1/\delta))$ times and outputting the
    median of the estimates. Thus, fixing any constant $\gamma \in (0,1)$,
    it suffices to show that every algorithm with failure
    probability $< \eps^{\frac{\gamma}{2} (d-1)}$ requires $(1/\epsilon)^{\Omega(d)}$ SQ queries.
    Let $\cB$ be a randomized SQ distance approximation algorithm for halfspaces under $\cN(0,I)$
    with failure probability $< \epsilon^{\frac{\gamma}{2}(d-1)}$.
    By a union bound over the set
    $\cF \cup \{ f_0 \}$ from the proof of \cref{thm:sq-lb-det},
    which consists of $(1/\epsilon)^{\frac{\gamma}{4}(d-1)} + 1
    < (1/\epsilon)^{\frac{\gamma}{2}(d-1)}$ functions,
    there is a random seed $r$ such that algorithm $\cB$ run with random
    seed $r$ outputs an accurate estimate on all of these functions. Let $\cA$ be a deterministic
    algorithm obtained from $\cB$ by using $r$ as advice instead of the random seed. Then $\cA$
    succeeds on $\cF \cup \{ f_0 \}$ and has the same query complexity as $\cB$.
    By \cref{thm:sq-lb-det}, we conclude that $\cA$ (and hence $\cB$) requires
    $(1/\epsilon)^{\Omega(d)}$ queries to $\Stat(\epsilon^{\Omega(d)})$.
\end{proof}

\bibliographystyle{alpha}
\bibliography{references}

\end{document}

%% file: header.tex
\usepackage[letterpaper, portrait, margin=1in]{geometry}
\usepackage[hypertexnames=false,colorlinks=true,linkcolor=blue,citecolor=ForestGreen]{hyperref}
\usepackage{algorithm}
\usepackage[noend]{algpseudocode}
\usepackage{url}
\usepackage{amsmath,amssymb,amsthm}
\usepackage{thmtools,thm-restate}
\usepackage[noabbrev,capitalise,nameinlink]{cleveref}
\usepackage{mathtools}
\usepackage{xspace}
\usepackage{verbatim}
\usepackage{mathrsfs}
\usepackage[usenames,dvipsnames,svgnames,table]{xcolor}
\usepackage{pgf}
\usepackage[dvipsnames]{xcolor}
\usepackage{todo}
\usepackage{tabularx}
\usepackage{enumitem}
\usepackage{derivative}
\usepackage{bm}
\usepackage{multirow}
\usepackage{diagbox}
\usepackage{nicematrix}
\usepackage{parskip}
\usepackage{adjustbox}
\usepackage{tikz}
\usepackage{tikz-3dplot}
\usepackage{transparent}
\usepackage{subcaption}

\usepackage{dsfont}

\newcommand*\ie{i.\kern.1em e., }
\newcommand*\eg{e.\kern.1em g., }
\newcommand*\cf{c.\kern.1em f.\ }
\newcommand*\almev{a.\kern.1em e.\ }

\theoremstyle{plain}
\newtheorem{theorem}{Theorem}[section]
\newtheorem{lemma}[theorem]{Lemma}
\newtheorem{fact}[theorem]{Fact}

\newtheorem{corollary}[theorem]{Corollary}

\newtheorem{conjecture}[theorem]{Conjecture}
\newtheorem{observation}[theorem]{Observation}
\newtheorem{assumption}[theorem]{Assumption}
\newtheorem{definition}[theorem]{Definition}
\newtheorem{problem}[theorem]{Problem}

\crefname{claim}{Claim}{Claims}
\crefname{fact}{Fact}{Facts}

\theoremstyle{definition}

\newtheorem{remark}[theorem]{Remark}

\theoremstyle{plain}

\newcommand{\ignore}[1]{}

\DeclareMathOperator{\poly}{poly}
\DeclareMathOperator{\polylog}{polylog}
\DeclareMathOperator{\sign}{sign}

\newcommand{\dist}{\mathsf{dist}}

\newcommand{\Ex}[1]{\bE \left[ #1 \right]}
\newcommand{\Exu}[2]{\underset{#1} \bE \left[ #2 \right] }

\renewcommand{\Pr}[1]{\bP \left[ #1 \right]} 
\newcommand{\Pru}[2]{\underset{ #1 }\bP \left[ #2 \right]}

\newcommand{\define}{\vcentcolon=}

\newcommand{\floor}[1]{\ensuremath{\left\lfloor #1 \right\rfloor}}
\newcommand{\ceil}[1]{\ensuremath{\left\lceil #1 \right\rceil}}

\DeclarePairedDelimiter{\abs}{\lvert}{\rvert}

\newcommand{\ind}[1]{\mathds{1} \left[ #1 \right] }

\newcommand{\zo}{\{0,1\}}
\newcommand{\pmset}{\{\pm 1\}}

\newcommand{\cA}{\ensuremath{\mathcal{A}}}
\newcommand{\cB}{\ensuremath{\mathcal{B}}}
\newcommand{\cC}{\ensuremath{\mathcal{C}}}
\newcommand{\cD}{\ensuremath{\mathcal{D}}}
\newcommand{\cE}{\ensuremath{\mathcal{E}}}
\newcommand{\cF}{\ensuremath{\mathcal{F}}}
\newcommand{\cG}{\ensuremath{\mathcal{G}}}
\newcommand{\cH}{\ensuremath{\mathcal{H}}}

\newcommand{\cN}{\ensuremath{\mathcal{N}}}
\newcommand{\cO}{\ensuremath{\mathcal{O}}}
\newcommand{\cP}{\ensuremath{\mathcal{P}}}
\newcommand{\cQ}{\ensuremath{\mathcal{Q}}}
\newcommand{\cR}{\ensuremath{\mathcal{R}}}
\newcommand{\cS}{\ensuremath{\mathcal{S}}}
\newcommand{\cT}{\ensuremath{\mathcal{T}}}

\newcommand{\cX}{\ensuremath{\mathcal{X}}}

\newcommand{\bE}{\ensuremath{\mathbb{E}}}

\newcommand{\bN}{\ensuremath{\mathbb{N}}}
\newcommand{\bP}{\ensuremath{\mathbb{P}}}
\newcommand{\bR}{\ensuremath{\mathbb{R}}}
\newcommand{\bS}{\ensuremath{\mathbb{S}}}

\newcommand{\bZ}{\ensuremath{\mathbb{Z}}}